\documentclass[runningheads]{llncs}

\spnewtheorem{numberedclaim}{Claim}{\bfseries}{\itshape}
\spnewtheorem*{proofidea}{Proof Idea}{\itshape}{\rmfamily}
\usepackage[T1]{fontenc}
\usepackage{graphicx}
\usepackage{hyperref}
\usepackage{amsmath,amssymb}
\usepackage{mathtools}
\usepackage[capitalize]{cleveref}

\crefname{numberedclaim}{Claim}{Claims}

\usepackage{enumerate}
\usepackage[dvipsnames]{xcolor}
\usepackage{soul}
\usepackage{tikz}
\usetikzlibrary{snakes,positioning,arrows.meta,arrows,
decorations.pathreplacing,calligraphy,petri,shapes.geometric}

\usepackage{subcaption}
\usepackage{tikz}
\usepackage{thm-restate}
\usetikzlibrary{automata, positioning, arrows.meta,fit}
\usepackage{graphicx}

\newcommand{\abs}[1]{|#1|}
\newcommand{\automata}[2]{#1[#2]}

\newcommand{\Cc}{\mathcal{C}}
\newcommand{\class}[1]{\textsf{#1}}
\newcommand{\Gg}{\mathcal{G}}
\newcommand{\lang}[1]{L(#1)}
\newcommand{\Nn}{\mathcal{N}}
\newcommand{\NN}{\mathbb{N}}
\newcommand{\norm}[1]{\lVert#1\rVert}
\newcommand{\Oh}{\mathcal{O}}
\newcommand{\powerset}[1]{\Pp(#1)}
\newcommand{\Pp}{\mathcal{P}}
\newcommand{\Rr}{\mathcal{R}}
\newcommand{\set}[1]{\{#1\}}
\newcommand{\sset}{\subseteq}
\newcommand{\tuple}[1]{(#1)}

\newcommand{\ZZ}{\mathbb{Z}}

\newcommand{\WC}{W}
\newcommand{\owner}{\mathit{owner}}

\newcommand{\+}[1]{\mathbb{#1}}

\newcommand{\N}{\+{N}}

\newcommand{\x}{\times}
\newcommand{\step}[2][]{\xrightarrow[#1]{#2}}

\tikzstyle{state} = [circle, draw, black, line width = 0.4mm, inner sep = 1.5mm]
\tikzstyle{transition} = [-{Stealth[width=1.5mm, length=2mm]}, black, line width = 0.4mm]
\tikzstyle{player1} = [regular polygon,regular polygon sides=4, draw, black, fill=white, line width = 0.4mm, inner sep = 0.2mm, minimum size = 6mm]
\tikzstyle{player0} = [circle, draw, black, line width = 0.4mm, inner sep = 0.2mm, minimum size = 4mm,fill=white] 
\usepackage{hyperref}
\hypersetup{bookmarksdepth=3,bookmarksopen=true,bookmarksnumbered=true}

\makeatletter
\AddToHook{cmd/appendix/before}{\def\cref@section@alias{appendix}}
\AddToHook{cmd/appendix/before}{\def\cref@subsection@alias{appendix}}
\AddToHook{cmd/appendix/before}{}
\makeatother

\usepackage[colorinlistoftodos,prependcaption,obeyFinal]{todonotes}

\usepackage{scalerel}
\usetikzlibrary{svg.path}
\definecolor{orcidlogocol}{HTML}{A6CE39}
\tikzset{
  orcidlogo/.pic={
    \fill[orcidlogocol] svg{M256,128c0,70.7-57.3,128-128,128C57.3,256,0,198.7,0,128C0,57.3,57.3,0,128,0C198.7,0,256,57.3,256,128z};
    \fill[white] svg{M86.3,186.2H70.9V79.1h15.4v48.4V186.2z}
                 svg{M108.9,79.1h41.6c39.6,0,57,28.3,57,53.6c0,27.5-21.5,53.6-56.8,53.6h-41.8V79.1z M124.3,172.4h24.5c34.9,0,42.9-26.5,42.9-39.7c0-21.5-13.7-39.7-43.7-39.7h-23.7V172.4z}
                 svg{M88.7,56.8c0,5.5-4.5,10.1-10.1,10.1c-5.6,0-10.1-4.6-10.1-10.1c0-5.6,4.5-10.1,10.1-10.1C84.2,46.7,88.7,51.3,88.7,56.8z};
  }
}
\tikzset{
  homelogo/.pic={
    \fill[black!60] svg{M18.69,73.37,59.18,32.86c2.14-2.14,2.41-2.23,4.63,0l40.38,40.51V114h-30V86.55a3.38,3.38,0,0,0-3.37-3.37H52.08a3.38,3.38,0,0,0-3.37,3.37V114h-30V73.37ZM60.17.88,0,57.38l14.84,7.79,42.5-42.86c3.64-3.66,3.68-3.74,7.29-.16l43.41,43,14.84-7.79L62.62.79c-1.08-1-1.24-1.13-2.45.09Z};
  }
}
\renewcommand\orcidID[1]{{\hypersetup{hidelinks}\href{https://orcid.org/#1}{\mbox{\scalerel*{
\begin{tikzpicture}[yscale=-1,transform shape]
\pic{orcidlogo};
\end{tikzpicture}
}{H}}}}}
\DeclareRobustCommand\urlll[1]{{\hypersetup{hidelinks}%
  \href{#1}{\mbox{\scalerel*{
    \begin{tikzpicture}[yscale=-1,transform shape]
      \pic{homelogo};
    \end{tikzpicture}
  }{H}}}%
}}

\title{History-Constrained Systems}

\titlerunning{History-Constrained Systems}

\author{Louwe B. Kuijer\inst{1}\orcidID{0000-0001-6696-9023}\,\urlll{https://sites.google.com/site/lbkuijer/}\and
        David Purser\inst{1}\orcidID{0000-0003-0394-1634}\,\urlll{https://www.davidpurser.net} \and \\[0.2cm]
        Henry Sinclair-Banks\inst{2,3}\orcidID{0000-0003-1653-4069}\,\urlll{http://henry.sinclair-banks.com} \and
        Patrick Totzke\inst{1}\orcidID{0000-0001-5274-8190}\,\urlll{https://cgi.csc.liv.ac.uk/~patrick/}}

\authorrunning{L.~B. Kuijer, D.~Purser, H.~Sinclair-Banks, and P.~Totzke}

\institute{University of Liverpool, United Kingdom
\and University of Warsaw, Poland 
\and Max Planck Institute for Software Systems (MPI-SWS), Kaiserslautern, Germany}

\begin{document}
\maketitle

\begin{abstract}
We study verification problems for history-constrained systems (HCS), a model of guarded computation that uses nested systems. An outer system describes the process architecture in which a sequence of actions represents the communication between sub-systems through a global bus. Actions are either permitted or blocked locally by guards; these guards read and decide based on the sequence of actions so far in the global bus. 

When HCS have both the outer systems and the local guard controllers modelled by finite automata, we show they have the same expressive power as regular languages and finite automata, but they are exponentially more succinct.
We also analyse games on this model, representing the interaction between environment and controller, and show that solving such games is \class{EXPTIME}-complete, where the lower bound already holds for reachability/safety games and the upper bound holds for any $\omega$-regular winning condition. 
Finally, we consider HCS with guards of greater expressive power,  Vector Addition Systems with States (VASS). We show that with deterministic coverability-VASS guards the reachability problem is \class{EXPSPACE}-complete, 
while with reachability-VASS the problem is undecidable.
 \end{abstract}


\section{Introduction}\label{sec:introduction}
\begin{figure*}[t!]
    \centering
\resizebox{0.99\textwidth}{!}{
\begin{tikzpicture}[
    >=stealth,
    auto,
    node distance=3cm,
    every state/.style={minimum size=20pt},
    every initial by arrow/.style={initial text={}},
    medstate/.style={state, fill=black!10, minimum size=14pt, inner sep=1pt},
    smallstate/.style={state, fill=blue!10, minimum size=12pt, inner sep=1pt},
    smalllabel/.style={
        font=\scriptsize,
    }
]

    \node[state, initial above, minimum size = 30pt] (q0) {$q_0$};

    \draw[transition] (q0) edge[loop left, olive!90] node{T1:*} (q0);
    \draw[transition] (q0) edge[loop right, teal!90] node{T2:*} (q0);

    \node[medstate, right=2.5cm of q0, yshift=-1cm] (b0) {};
    \node[medstate, right=3cm of b0] (b1) {};
    \node[medstate, above=1cm of b1, xshift=0.5cm] (b2) {};
    \node[medstate, below left=0.25cm of b2,xshift=-2cm] (b3) {};
    \node[medstate, above left=1.5cm of b0, xshift=1cm] (b4) {};

    \draw[transition] (b0) edge[above, red!50] node{T2:turn-green} (b1);
    \draw[transition] (b1) edge[gray!50] node{} (b2);
    \draw[transition] (b2) edge[gray!50] node{} (b3);
    \draw[transition] (b3) edge[gray!50] node{} (b4);
    \draw[transition] (b4) edge[gray!50] node{} (b0);

    \node[smallstate, fill=red!10, initial, below right=0.8cm of b0] (s0) {};
    \node[smallstate, fill=red!10, accepting, right=1.5cm of s0] (s1) {};

    \draw[transition] (s0) edge[bend left] node[smalllabel]{T1:turn-red} (s1);
    \draw[transition] (s1) edge[bend left] node[smalllabel,align=center](labelT1turngreen){T1:turn-green\\T1:turn-orange} (s0);

    \node[draw, red!50, rounded corners, fit=(s0)(s1)(labelT1turngreen), inner sep=8pt,xshift=-0.18cm,yshift=0.25cm] (box23) {};

    \draw[transition] (s0) edge[loop above] (s0);
    \draw[transition] (s1) edge[loop above] (s1);

 	\node[draw, teal, rounded corners, fit=(b0)(b1)(b2)(b3)(b4)(box23), inner sep=5pt] (box23) {};

	\node () at (5, 1.2) {Traffic Light 2 controller};

    \node[medstate, left=6cm of q0, yshift=-1cm] (b0) {};
    \node[medstate, right=3cm of b0] (b1) {};
    \node[medstate, above=1cm of b1, xshift=0.5cm] (b2) {};
    \node[medstate, below left=0.25cm of b2,xshift=-2cm] (b3) {};
    \node[medstate, above left=1.5cm of b0, xshift=1cm] (b4) {};

    \draw[transition] (b0) edge[above, blue!50] node{T1:turn-green} (b1);
    \draw[transition] (b1) edge[gray!50] node{} (b2);
    \draw[transition] (b2) edge[gray!50] node{} (b3);
    \draw[transition] (b3) edge[gray!50] node{} (b4);
    \draw[transition] (b4) edge[gray!50] node{} (b0);

    \node[smallstate, initial, below right=0.8cm of b0,yshift=0cm] (s0) {};
    \node[smallstate, accepting, right=1.5cm of s0] (s1) {};

    \draw[transition] (s0) edge[bend left] node[smalllabel](labelT2turnred){T2:turn-red} (s1);
    \draw[transition] (s1) edge[bend left] node[smalllabel, align=center](labelT2turngreen){T2:turn-green\\T2:turn-orange} (s0);

    \node[draw, blue!50, rounded corners, fit=(s0)(s1)(labelT2turngreen), inner sep=8pt,xshift=-0.18cm,yshift=0.25cm] (box23) {};

    \draw[transition] (s0) edge[loop above] (s0);
    \draw[transition] (s1) edge[loop above] (s1);

    \node[draw, olive, rounded corners, fit=(b0)(b1)(b2)(b3)(b4)(box23), inner sep=5pt] (box23) {};

	\node () at (-5, 1.2) {Traffic Light 1 controller};

\end{tikzpicture}
}
    \caption{Three layer HCS traffic light model. The outer system is a single state to hold each controller; ``T1:*'' represents all actions for traffic light 1. The middle layer represents the controller of each individual traffic light (the full mechanism of the traffic light controller  (grey edges) is not depicted). The innermost system guards the ``turn green'' commands to ensure that the other set of lights have turned to red. The system enforces safety by ensuring that whenever a traffic light attempts to turn green, the other traffic light has turned to red. Unlabelled self-loops represent all actions other than those specified. 
    }
    \label{fig:traffic}
\end{figure*}
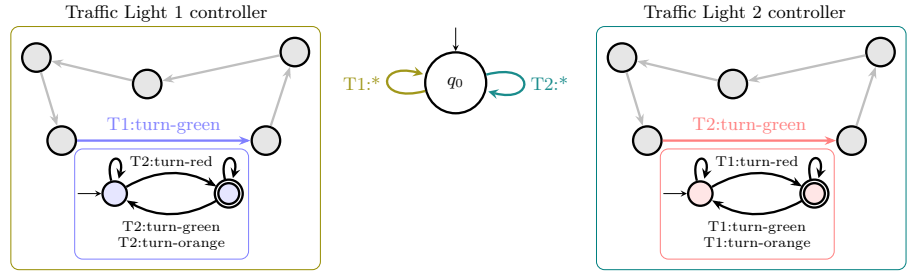

We study distributed systems where independent components communicate using a global databus --- a shared record of actions taking place in the system. The subsystems act and make decisions based on their own internal state and information from the global communication channel, but are isolated from the internal configuration of other subsystems. We introduce a model of guarded computation called history-constrained systems (HCS) that capture such distributed workflows. We then study the expressivity of this model as well as the computational cost of formally verifying such models.

Formally, HCS consist of an outer state-machine, whose transitions are labelled with global actions and which are guarded by languages over such actions whereby a transition is admissible if the whole sequence of prior actions is in the (guard) language.
For practical purposes, we will assume that the guard languages are themselves provided in the form of automata.
Our systems can thus be seen as nested automata: an outer system and inner automata acting as guards on the transitions. The outer system models the overall architecture of the large-scale task, system, or process. Every action taken in the outer system is recorded in a global history (the bus). Each edge in the system may have a guard language, encoded by an automaton, which observes the event history so far and decides whether the transition is enabled at a given moment.

HCS are safe by construction: unsafe behaviour expressed as properties of the history are easily encoded into the guard conditions. As an example, consider a traffic-light controller that permits a ``turn-green'' action only after observing that perpendicular lights have executed a ``turn-red'' action, without requiring access to their internal states (see~\cref{fig:traffic}). Writing rules over event histories is usually simpler than writing imperative rules over joint states. Other examples which can easily be encoded include:

\vspace{-\topsep}\begin{itemize}

\item Prerequisites/Cross-service dependencies ``Service X may only perform action A if service Y has previously sent confirmation B'',
\item Rate limiting %
``Permit X only if 10 steps have passed since event Y'', or
\item Safety guards %
``Event X cannot happen if Event Y has happened''.
\end{itemize}

Even with this ``safe by construction'' approach, it remains necessary to verify that the model can achieve its intended purpose under these guards: can the automaton still achieve its goal (reachability)? Is there a controller that can act against an adversarial environment? We investigate such questions for HCS.

The expressive power of the guard language determines the properties of the history upon which the decisions can be made (patterns, counts, timestamps, presence/absence of events) and the complexity or decidability of the verification problems. For example, if we use context-free languages as guards, then these problems become undecidable (see~\cref{subsec:context_free}).

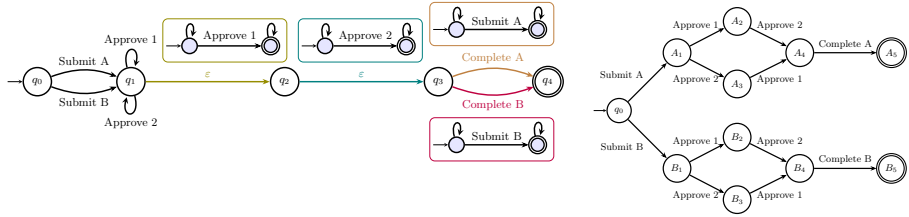
\begin{figure*}[t!]
    \centering
    \begin{subfigure}[t]{0.63\textwidth}
        \centering
        
			\resizebox{0.99\textwidth}{!}{
			\begin{tikzpicture}[
			    >=stealth,
			    auto,
			    node distance=3cm,
			    every state/.style={minimum size=20pt},
			    every initial by arrow/.style={initial text={}},
			    smallstate/.style={state, fill=blue!10, minimum size=12pt, inner sep=1pt}
			]

			    \node[state, initial] (q0) {$q_0$};
			    \node[state, right=1.75cm of q0] (q1) {$q_1$};
			    \node[state, right=3.4cm of q1] (q2) {$q_2$};
			    \node[state, right=3.4cm of q2] (q3) {$q_3$};
			    \node[state, accepting, right of=q3] (q4) {$q_4$};

			    \draw[transition] (q0) edge[above,bend left=20] node{Submit A} (q1);
			    \draw[transition] (q0) edge[below,bend right=20] node{Submit B} (q1);

			    \draw[transition] (q1) edge[loop above] node{Approve 1} (q1);
			    \draw[transition] (q1) edge[loop below] node{Approve 2} (q1);
			    \draw[transition] (q1) edge[above,olive] node{$\varepsilon$} (q2);
			    \draw[transition] (q2) edge[above,teal] node{$\varepsilon$} (q3);
			    \draw[transition] (q3) edge[above,bend left=20,brown] node{Complete A} (q4);
			    \draw[transition] (q3) edge[below,bend right=20,purple] node{Complete B} (q4);

			    \path (q2) -- (q3) coordinate[midway] (mid23);
			    \path (q1) -- (q2) coordinate[midway] (mid12);

			    \node[smallstate, initial, above=0.75cm of mid12, xshift=-0.5cm] (r0) {};
			    \node[smallstate, accepting, right=1.75cm of r0] (r1) {};

			    \node[draw, olive, rounded corners, fit=(r0)(r1), inner sep=10pt,xshift=-0.15cm,yshift=0.15cm] (box23) {};

			    \draw[transition] (r0) edge[above] node{Approve 1} (r1);
			    \draw[transition] (r0) edge[loop above] node{} (r0);
			    \draw[transition] (r1) edge[loop above] node{} (r1);

			    \node[smallstate, initial, above=0.75cm of mid23, xshift=-1cm] (s0) {};
			    \node[smallstate, accepting, right=1.75cm of s0] (s1) {};

			    \node[draw, teal, rounded corners, fit=(s0)(s1), inner sep=10pt,xshift=-0.15cm,yshift=0.15cm] (box23) {};

			    \draw[transition] (s0) edge[above] node{Approve 2} (s1);
			    \draw[transition] (s0) edge[loop above] node{} (s0);
			    \draw[transition] (s1) edge[loop above] node{} (s1);

			    \path (q3) -- (q4) coordinate[midway] (mid34);

			    \node[smallstate, initial, above=1.2cm of mid34, xshift=-1cm] (u0) {};
			    \node[smallstate, accepting, right=1.75cm of u0] (u1) {};

			    \node[draw, brown, rounded corners, fit=(u0)(u1), inner sep=10pt,xshift=-0.15cm,yshift=0.15cm] (box23) {};

			    \draw[transition] (u0) edge[above] node{Submit A} (u1);
			    \draw[transition] (u0) edge[loop above] node{} (u0);
			    \draw[transition] (u1) edge[loop above] node{} (u1);

			    \node[smallstate, initial, below=1.5cm of mid34, xshift=-1cm] (b0) {};
			    \node[smallstate, accepting, right=1.75cm of b0] (b1) {};

			    \node[draw, purple, rounded corners, fit=(b0)(b1), inner sep=10pt,xshift=-0.15cm,yshift=0.15cm] (box23) {};

			    \draw[transition] (b0) edge[above] node{Submit B} (b1);
			    \draw[transition] (b0) edge[loop above] node{} (b0);
			    \draw[transition] (b1) edge[loop above] node{} (b1);

			    \node[below=3cm of q0] (asdf){}; %

			\end{tikzpicture}
			}
    \end{subfigure}%
    ~ 
    \begin{subfigure}[t]{0.36\textwidth}
        \centering
        \resizebox{0.99\textwidth}{!}{
			\begin{tikzpicture}[
			    >=stealth,
			    auto,
			    node distance=2.5cm,
			    every initial by arrow/.style={initial text={}},
			    every state/.style={minimum size=20pt}
			]

			    \node[state, initial] (q0) {$q_0$};
			    \node[state, above right=1.75cm of q0] (q1) {$A_1$};

			    \node[state, right=1cm of q1, yshift=1cm] (q2) {$A_2$};
			    \node[state, right=1cm of q1, yshift=-1cm] (q3) {$A_3$};

			    \node[state, right=3cm of q1] (q4) {$A_4$};
			    \node[state, accepting, right=2cm of q4] (q6) {$A_5$};

			    \draw[transition] (q0) edge[above left ] node{Submit A} (q1);

			    \draw[transition] (q1) edge[above left] node[xshift=0.5cm,yshift=0.1cm]{Approve 1} (q2);
			    \draw[transition] (q1) edge[below left] node[xshift=0.5cm,yshift=-0.1cm]{Approve 2} (q3);

			    \draw[transition] (q3) edge[below right] node[xshift=-0.5cm,yshift=-0.1cm]{Approve 1} (q4);
			    \draw[transition] (q2) edge[above right] node[xshift=-0.5cm,yshift=0.1cm]{Approve 2} (q4);

			    \draw[transition] (q4) edge[above] node{Complete A} (q6);

			    \node[state, below right=1.75cm of q0] (t1) {$B_1$};

			    \node[state, right=1cm of t1, yshift=1cm] (t2) {$B_2$};
			    \node[state, right=1cm of t1, yshift=-1cm] (t3) {$B_3$};

			    \node[state, right=3cm of t1] (t4) {$B_4$};
			    \node[state, accepting, right=2cm of t4] (t6) {$B_5$};

			    \draw[transition] (q0) edge[below left ] node{Submit B} (t1);

			    \draw[transition] (t1) edge[above left] node[xshift=0.5cm,yshift=0.1cm]{Approve 1} (t2);
			    \draw[transition] (t1) edge[below left] node[xshift=0.5cm,yshift=-0.1cm]{Approve 2} (t3);

			    \draw[transition] (t3) edge[below right] node[xshift=-0.5cm,yshift=-0.1cm]{Approve 1} (t4);
			    \draw[transition] (t2) edge[above right] node[xshift=-0.5cm,yshift=0.1cm]{Approve 2} (t4);

			    \draw[transition] (t4) edge[above] node{Complete B} (t6);

			\end{tikzpicture}
			}

    \end{subfigure}
    \caption{Four-eyes approval workflow represented as both an HCS (left) and as a finite automaton (right). Approvals arrive in an unknown order, the finite automaton must explicitly encode this in the state, while the HCS can store this information on just the relevant edges. Three approvers would enlarge the finite automaton significantly, while the HCS needs just one more guarded transition.  Unlabelled self-loops represent all actions other than those specified. 
    }
    \label{fig:foureyes}
\end{figure*}

\subsubsection{Contributions}
Besides introducing the model of history-constrained systems, we make contributions in two directions: when  guards are regular languages and when  guards are represented by automata with (limited) access to counters.

\paragraph{Regular guards.}
HCS with regular guards have transitions that are each controlled by a deterministic finite automaton (DFA) or a non-deterministic finite automaton (NFA).
We show that their expressivity is the same as that of DFAs (\cref{sec:expressivity}) but even deterministic HCS can be exponentially more succinct than NFAs (\cref{sec:succinctness}). To illustrate, consider the conditional workflow in \cref{fig:foureyes}: allowing conditions to be encoded locally avoids needing to encode information in the state space.

In~\cref{sec:games} we consider games on HCS with regular guards and show that determining the existence of a winning strategy is \class{EXPTIME}-complete (\cref{thm:exptime-completeness}).
Note that this is a non-trivial bound, since such HCS are exponentially more succinct than DFAs, and games on the latter are \class{PSPACE}-complete (and hence we would trivially obtain an \class{EXPSPACE} upper bound). However, we are able to reduce this by carefully composing the parity game algorithm with the on-the-fly construction of the state space. The lower bound is one of our main contributions, by reducing from countdown games.

\paragraph{Guards with counters.}
In \cref{sec:vass} we consider the natural extension to allow guards to use counters. Such guards would be able to directly extend the example in \cref{fig:foureyes} with ``majority approval'', where more people approve than reject. However, zero tests lead to undecidability, and so restrictions are needed. 

We consider restricted guards with counters using Vector Addition Systems with States (VASS). VASS have one or more non-negative integer counters which can be incremented or decremented (but can neither be reset nor zero-tested). VASS have two common acceptance conditions: Cover-VASS ask whether a final state is reachable but make no requirement on the final value of the counter (only that they remained non-negative), while Reach-VASS insist that the counters are finally exactly zero (or some other fixed choice).

Our results reveal an interesting subtlety depending on the acceptance conditions. HCS with deterministic Cover-VASS guards match exactly the expressive power of Cover-VASS, and the reachability problem is decidable in \class{EXPSPACE}.

However, HCS with deterministic \hbox{Reach-VASS} guards are more expressive than just Reach-VASS. While reachability in Reach-VASS is decidable (and recently shown to be \class{Ackermann}-complete~\cite{czerwinski2022reachability,leroux2022reachability}), it is undecidable in HCS with deterministic Reach-VASS guards. %

\paragraph{Published version.} 
This is the full version of the (open-access) paper in Formal Methods 2026~\cite{KuijerPST26}.

\subsubsection{Related work}
Distributed systems can broadly be classified by how individual components communicate: one-to-one (aka rendezvous) \cite{milner1980calculus} and broadcasting (aka shared memory) \cite{DBLP:journals/scp/Prasad95}. %
Early automata models for broadcast communication include I/O automata \cite{ioautomata} and team automata \cite{teamautomata} where components communicate via shared actions
but compared to our model, they do not expose the whole action history to guards like HCS allow.

Our work follows the canon of models with guarded computation, where the availability of local state changes depends on %
the history of the process so far.
Examples of such models are Timed automata \cite{AD1994} which allow for inequality constraints on real-valued clocks; Petri Nets / Vector addition systems \cite{vass-games-undec}, where counters must remain non-negative.
The structure of HCS as an outer system with inner ones representing the guards can be found in the \emph{Nets in Nets} model \cite{DBLP:conf/ac/Valk03,10.1007/978-3-642-02424-5_15}, where Petri net tokens themselves represent identifiable Petri nets that can move and synchronise with the outer one in a rendez-vous fashion. That model is convenient for multi-agent systems, where it can model agent mobility. In HCS, individual components do not move but synchronise via broadcast, and we also consider different kinds of guard automata.

Our model can be used to succinctly model arenas for games on graphs,
the mathematical foundation for model checking and reactive synthesis
\cite{GAMES2002,GOG,PR1989,P1977}.
Related structures are \emph{Temporal graphs}, \cite{HJ19,M2015} which extend graphs with discrete and global time and can specify for each edge at which times it can be traversed. Games on temporal graphs, succinctly encoded in existential Presburger Arithmetic, have recently been considered in \cite{AustinBT24,ABMT2025}. Games on HCS, as done in \cref{sec:games} here, are essentially a generalisation of games on temporal graphs from linear to branching-time.

\section{Preliminaries}\label{sec:preliminaries}

Let $X$ be a finite set, and $\powerset{X} \coloneqq \set{X' : X' \sset X}$ be the power set of $X$. 
$\varepsilon$ is the empty word of length 0. 
Let $\vec{v}_1 \in \ZZ^{d_1}$ and $\vec{v}_2 \in \ZZ^{d_2}$. We define the vector $(\vec{v}_1, \vec{v}_2) \coloneqq (\vec{v}_1[1], \ldots, \vec{v}_1[d_1], \vec{v}_2[1], \ldots, \vec{v}_2[d_2]) \in \ZZ^{d_1 + d_2}$.

\subsection{History-Constrained Systems}

A \emph{History-constrained system} (HCS) is a tuple $A = \tuple{\Sigma, U, \Gg}$ where $\Sigma$ is a finite alphabet, $U$ is a labelled transition system over the alphabet $\Sigma$ with transitions $T$, and $\Gg = \set{L_t \sset \Sigma^* : t \in T}$ are the guards. A guard $L_t$ is \emph{trivial} if $L_t = \Sigma^*$. 

We say that a sequence $(q_0, a_1, q_1, a_2, q_2, \dots, q_{n-1}, a_n, q_n)$ is a run of $A$ if 
\begin{enumerate}[(1)]
	\item $q_0$ is the starting state of $U$; and
	\item for every $1 \leq i \leq n$, $t = (q_{i-1}, a_i, q_i)$ is a transition in $U$ and \hbox{$a_1 a_2 \ldots a_{i-1} \in L_t$}.
\end{enumerate}

Note that the runs of $A$ are a subset of the runs of $U$. When $U$ is an automaton with an accepting condition, we say a run is accepting in $A$ if the same run is accepting in $U$.
In this case we say $w = a_1 a_2 \ldots a_n$ is accepted by $A$ (i.e. $w \in L(A)$) if there is an accepting run on $w$, and $L(A)$ is the language of $A$.

In other words, while reading the word $w = a_1 a_2 \ldots a_n$, whenever a transition $t$ is taken to read $a_i$ (the $i$-th letter of $w$), it must be true that the current prefix read so far is accepted by the guard of the $i$-th transition, $a_1 a_2 \ldots a_{i-1} \in L_t$.

We say that an HCS $(\Sigma, U, \Gg)$ is deterministic or non-deterministic if $U$ is deterministic or non-deterministic, respectively.

\subsection{Finite Automata}
A \emph{non-deterministic finite automaton} (NFA) is a 5-tuple $A = (\Sigma, Q, \delta, q_0, F)$ where $\Sigma$ is a finite alphabet, $Q$ is a finite set of states, $\delta:Q\times \Sigma\cup\{\epsilon\}\to \powerset{Q}$ is a transition function, $q_0\in Q$ is an initial state and $F\subseteq Q$ is the set of accepting states.  $A$ is \emph{deterministic} (DFA) if $|\delta(q,\sigma)|\le 1$ for all $q\in Q$ and $\sigma \in \Sigma$ and $\delta(q,\epsilon)=\emptyset$. $A$ accepts  $w\in\Sigma^*$ if there is a run from $q_0$ to $F$ labelled by $w$.
The size of $A$ is $\norm{A} \coloneqq \abs{Q} + \sum_{q \in Q, \sigma \in \Sigma\cup\set{\varepsilon}} \abs{\delta(q, \sigma)}$.

\subsection{HCS Specified with Finite Automata}

The languages guarding transitions can be specified however one wishes. 
As we alluded to in the introduction, we will consider the case that both the underlying automaton $U$ and the transition guards are specified using finitely representable automata. Where the guard languages are regular languages we can specify the guards using  DFAs or NFAs. 
We denote how the underlying automaton and guard languages are specified using \automata{X}{Y} for an HCS whose underlying automaton is of type X and guard languages are specified as automata of type Y. 
The variants of HCS that we study include:
\vspace{-\topsep}
\begin{itemize}
	\item \automata{DFA}{DFA} are deterministic HCS with guards  specified by DFAs, 
	\item \automata{NFA}{DFA} are non-deterministic HCS with guards  specified by DFAs,
	\item \automata{DFA}{VASS} are deterministic HCS with guards  specified by vector addition systems with states, and 
	\item \automata{DFA}{PDA} are deterministic HCS with guards  specified by pushdown automata.
\end{itemize}

\subsection{Emptiness and Reachability}

Given an HCS $A= (\Sigma, U, \mathcal{G})$, the \emph{emptiness} problem asks if $L(A) = \emptyset$. That is, no word has an \emph{accepting} run from the initial state consistent with the guards on every transition. 

Where the accepting condition of $U$ is specified by accepting (or final) states $F$, we say $F$ is \emph{reachable} if $L\ne \emptyset$, and \emph{unreachable} otherwise. For such systems, the emptiness problem and reachability problem coincide.

\subsection{Examples}

Consider the two examples of \automata{DFA}{DFA} over $\Sigma = \set{a,b}$ depicted in~\cref{fig:examples}  which share the guards \textcolor{Red}{$L_1$}, \textcolor{Green}{$L_2$}, and \textcolor{Blue}{$L_3$} defined as:
\begin{itemize}
	\item \textcolor{Red}{$L_1 = (ab)^* + a(ba)^*$} (Words of alternating `$a$'s and `$b$'s.)
	\item \textcolor{Green}{$L_2 = \set{w \in \set{a, b}^* : \abs{w} \in 3\cdot\NN}$} (Words whose length is a multiple of three.)
	\item \textcolor{Blue}{$L_3 = \set{a,b}^* aa \set{a,b}^*$} (Words containing two consecutive `$a$'s.)
\end{itemize}

The language of the \automata{DFA}{DFA} with the underlying automaton $U_1$ is empty. 
		In order to pass the first guard, zero `$b$'s have to be read. 
		Thus, after the \textcolor{Red}{$L_1$} guard is passed, the prefix of the word is just `$a$'.
		In order to pass the \textcolor{Green}{$L_2$} guard, at least two `$b$' must be read using the second cycle.
		Hence, once the \textcolor{Green}{$L_2$} guard has been passed the prefix of the word cannot contain two consecutive `$a$'s.
		This means that the \textcolor{Blue}{$L_3$} guard cannot be passed, and so the language is empty. 
		The language of the \automata{DFA}{DFA} with the underlying automaton $U_2$ is not empty.
		The word `$baaba$' is accepted.
		By traversing the lower branch, multiple `$a$'s can be read before passing the \textcolor{Green}{$L_2$} guard which allows for the later passing of the \textcolor{Blue}{$L_3$} guard. 

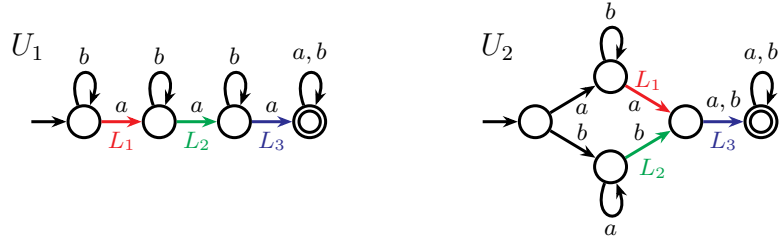
\begin{figure}[t]
	\centering
\begin{tikzpicture}
	\node at (-0.75, 1) {\large $U_1$};

	\node[state] (q1) at (0,0) {};
	\node[state] (q2) at (1,0) {};
	\node[state] (q3) at (2,0) {};
	\node[state, inner sep = 1mm] (q4) at (3,0) {};
	\node[state] (q4) at (3,0) {};

	\draw[transition] (-0.7,0) -- (q1);

	\draw[transition, Red] (q1) -- node[above]{\textcolor{black}{$a$}} node[below]{$L_1$} (q2);
	\draw[transition, Green] (q2) -- node[above]{\textcolor{black}{$a$}} node[below]{$L_2$} (q3);
	\draw[transition, Blue] (q3) -- node[above]{\textcolor{black}{$a$}} node[below]{$L_3$} (q4);

	\draw[transition] (q1) edge[loop below, in = 70, out = 110, distance = 6mm] node[above] {$b$} (q1);
	\draw[transition] (q2) edge[loop below, in = 70, out = 110, distance = 6mm] node[above] {$b$} (q2);
	\draw[transition] (q3) edge[loop below, in = 70, out = 110, distance = 6mm] node[above] {$b$} (q3);
	\draw[transition] (q4) edge[loop below, in = 70, out = 110, distance = 6mm] node[above] {$a,b$} (q4);

	\begin{scope}[xshift=1cm]
	\node at (4.5, 1) {\large $U_2$};

	\node[state] (p1) at (5,0) {};
	\node[state] (p2) at (6,0.6) {};
	\node[state] (p3) at (6,-0.6) {};
	\node[state] (p4) at (7,0) {};
	\node[state, inner sep = 1mm] (p5) at (8,0) {};
	\node[state] (p5) at (8,0) {};

	\draw[transition] (4.3, 0) -- (p1);

	\draw[transition] (p1) -- node[below,pos=0.7]{$a$} (p2);
	\draw[transition] (p1) -- node[above,pos=0.7] {$b$} (p3);
	\draw[transition, Red] (p2) -- node[above]{$L_1$} node[below, pos=0.2]{\textcolor{black}{$a$}}  (p4);
	\draw[transition, Green] (p3) -- node[above, pos=0.3]{\textcolor{black}{$b$}} node[below=0.1, pos=0.6]{$L_2$} (p4);
	\draw[transition, Blue] (p4) -- node[above] {\textcolor{black}{$a,b$}} node[below]{$L_3$} (p5);

	\draw[transition] (p2) edge[loop below, in = 70, out = 110, distance = 6mm] node[above] {$b$} (p2);
	\draw[transition] (p3) edge[loop below, in = -70, out = -110, distance = 6mm] node[below] {$a$} (p3);
	\draw[transition] (p5) edge[loop below, in = 70, out = 110, distance = 6mm] node[above] {$a,b$} (p5);
	\end{scope}
\end{tikzpicture}
 	\caption{Examples $U_1$ and $U_2$ sharing guards \textcolor{Red}{$L_1$}, \textcolor{Green}{$L_2$}, and \textcolor{Blue}{$L_3$}.	Whenever we do not specify a guard, then one should assume we use the trivial guard $L_t = \Sigma^*$.
	}
	\label{fig:examples}
\end{figure}

\section{Expressivity of HCS}\label{sec:expressivity}
In this section we show that for finite automata, adding regular guards does not increase the expressivity (although in \cref{sec:succinctness} we show they are more succinct). We then show that, in general, HCS can encode the intersection languages of their guard, with the consequence that context-free guards are too powerful.

\subsection{The Expressivity of Regular Guards}
We will first prove that finite automata with finite automata guards (\automata{NFA}{NFA}) recognise the regular languages (\cref{thm:exponential_size}).
For this, we will construct a single DFA by converting each of the guard NFAs into equivalent DFAs and then convert the underlying NFA into a DFA whilst taking the product of the underlying NFA with the guard DFAs.

\begin{restatable}{theorem}{thmregexpsize}\label{thm:exponential_size}
	Given an \automata{NFA}{NFA} $A$, we can construct an exponential size DFA $D$ such that $\lang{D} = \lang{A}$.
\end{restatable}

\begin{proofidea}[Full proof in~\cref{appendix:regularexpressivity}]
	The proof uses a careful mix of the powerset and product constructions. 
	If $A$ were a \automata{DFA}{DFA} then we would just take the product of the underlying DFA with each of the guard DFAs (which are assumed to be complete), only keeping transitions if the relevant guard is in an accepting state.

	Now, let $A$ be an \automata{NFA}{NFA} comprising an underlying NFA $U$ with $n$ states and NFA guards $G_1,\dots,G_m$ with $n_1,\dots,n_m$ states respectively; in total $A$ has $n+n_1+\dots+n_m$ states. 
	First, we directly transform $G_1,\dots,G_m$ to DFAs using the standard powerset construction, each with at most $2^{n_i}$ states. 

	It remains to determinise the underlying automaton $U$. However, we must carefully avoid all possible products of guards. Instead, we compute a simultaneous powerset construction on $U$ with a product construction on the guard DFAs. Each state of $D$ takes the form $(Q',q_1,\dots,q_m)$ where $Q'$ is a subset of the states of the underlying automata and each $q_i$ is a state of the $i$-th guard DFA. 
	During this construction, we keep only transitions where an appropriate subset of the guards are in an accepting state. In total, $D$ has at most $2^n\cdot 2^{n_1}\cdot\ldots\cdot2^{n_m} = 2^{n+n_1+\dots+n_m}$ states. 
\end{proofidea}

\begin{corollary}
	\automata{DFA}{DFA}, \automata{DFA}{NFA}, \automata{NFA}{DFA}, \automata{NFA}{NFA} all recognise exactly the regular languages.
\end{corollary}

\subsection{Computing Intersections with HCS}

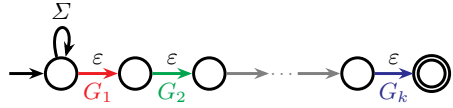
\begin{figure}[t]
\begin{center}
			\resizebox{6cm}{!}{%
\begin{tikzpicture}
	\node[state] (q1) at (0,0) {};
	\node[state] (q2) at (1,0) {};
	\node[state] (q3) at (2,0) {};
	\node[state] (q4) at (4,0) {};
	\node[state,accepting] (q5) at (5,0) {};

	\node[gray,inner sep=0.5, scale=0.8] (dots) at (3,0) {$\dots$};

	\draw[transition] (-0.7,0) -- (q1);

	\draw[transition, Red] (q1) -- node[above]{\textcolor{black}{$\varepsilon$}} node[below]{$G_1$} (q2);
	\draw[transition, Green] (q2) -- node[above]{\textcolor{black}{$\varepsilon$}} node[below]{$G_2$} (q3);

	\draw[transition, gray] (q3) -- (dots);

	\draw[transition, gray] (dots) -- (q4);
	\draw[transition, Blue] (q4) -- node[above]{\textcolor{black}{$\varepsilon$}} node[below]{$G_k$} (q5);

	\draw[transition] (q1) edge[loop above, distance = 15] node {$\Sigma$} (q1);
\end{tikzpicture} %
}
		\end{center}
\caption{\automata{NFA}{X} for the intersection of $G_1,\dots, G_k$}
\label{fig:intersection}
\end{figure}

\begin{proposition}\label{pro:intersection-to-NFA}
	Given automata $G_1, \ldots, G_k$ of type $X$ over a common alphabet $\Sigma$, one can construct an HCS \automata{NFA}{X} $A$ of size $\Oh(k) + \sum_{i=1}^k \abs{G_i}$ such that $L(A) = \bigcap_{i=1}^k \lang{G_i}$.
\end{proposition}
\begin{proof}
	We construct the \automata{NFA}{X} $A = (\Sigma, U, \Gg)$. Let $U = (\Sigma, Q, \delta, q_0, F)$ be the underlying NFA with $Q = \set{q_0, q_1, \ldots, q_k, q_{k+1}}$, $\delta = \set{(q_0, \sigma, q_0) : \sigma \in \Sigma} \cup \set{(q_{i-1}, \varepsilon, q_i) : 1 \leq i \leq k}$, and $F = \set{q_{k+1}}$ as depicted in~\cref{fig:intersection}.
	The \automata{NFA}{X} $A = (\Sigma, U, \Gg)$ has the guards: on the $q_0$-transitions there are no guards and, for every $1 \leq i \leq k$, the transition $(q_{i-1}, \varepsilon, q_i)$ has guard $G_i$.
	It is straightforward that, for every $w \in \Sigma^*$, $w$ is accepted by $A$ iff $w \in \lang{G_i}$ for every $1 \leq i \leq k$. 
\qed\end{proof}

Similarly, if we want a DFA we can replace the $\varepsilon$ transitions in~\cref{fig:intersection} with a new `$\$$' symbol (and update the guards to ignore them), entailing the following (full proof in~\cref{appendix:dfaintersection}):

\begin{restatable}{proposition}{intersectiondfa}\label{pro:intersection-to-DFA}
	Given automata $G_1, \ldots, G_k$ of type $X$ over a common alphabet $\Sigma$, one can construct a \automata{DFA}{X} $A$ of size $\Oh(k^2) + \sum_{i=1}^k \abs{G_i}$ such that $L(A) = \set{w \$^k : w \in \bigcap_{i=1}^k \lang{G_i}}$ where `$\,\$$' is a symbol not in $\Sigma$.
\end{restatable}

\subsection{Context-free Guards are Too Powerful}
\label{subsec:context_free}
We now observe that context-free languages, recognised by pushdown automata (PDA), are too powerful to be used as guards. Recall that it is undecidable whether the intersection of two context-free languages is empty~\cite[Theorem 9.22]{hopcroft2001introduction}. Thus as a corollary of \cref{pro:intersection-to-DFA} we have:
\begin{corollary}\label{cor:pda-are-undecidable}
	Emptiness of \automata{DFA}{PDA} is undecidable.
\end{corollary}

In fact, in \cref{sec:vass}, we show that emptiness is already undecidable for a subcase of pushdown automata: vector addition systems with states with just one counter under reachability acceptance conditions.

\section{Succinctness of HCS with Regular Guards}\label{sec:succinctness}

For any $\ell \in \NN$, let $C_\ell$ be the $\ell$-state DFA over the alphabet $\Sigma = \set{a}$ with the language $L(C_\ell) = \set{a^{\ell n} : n \in \NN}$. 
In other words, $C_\ell$ is the DFA that is a simple cycle of $a$-transitions of length $\ell$. 
We will use the pumping lemma for finite automata~\cite[Theorem 4.1]{hopcroft2001introduction} to prove the following claim.

\begin{restatable}{lemma}{claimexpna}\label{clm:exp-nfa}
	Let $p_1, \ldots, p_k$ be the first $k$ primes.
	Any NFA $N$ with the language $\lang{N} = \set{w \$^k : w \in \bigcap_{i=1}^k L(C_{p_i})}$ has size at least $2^k$.
\end{restatable}
\begin{proofidea}[Full proof in \cref{appendix:succinctness}]
The product of the first $k$ primes is at least $2^k$. Since the difference in length between any two accepting words is at least $2^k$, there must be a cycle-free path of this length (or else a cycle could be depumped), entailing at least $2^k$ distinct states.\qed
\end{proofidea}

We conclude this subsection by observing that \automata{DFA}{DFA} can be exponentially more succinct than NFAs. 
\cref{cor:succinctness} follows from~\cref{pro:intersection-to-DFA} and Claim~\ref{clm:exp-nfa}. See the full proof in~\cref{appendix:succinctness}.

\begin{restatable}{corollary}{succinctness}\label{cor:succinctness}
	There exists an infinite family of \automata{DFA}{DFA}s $\set{A_k : k \in \NN}$ such that, for every $i \in \NN$ and for every NFA $N$ such that $\lang{N} = \lang{A}$, it is true that $\abs{N} \geq 2^{C\abs{A_k}^{\frac{1}{3}}}$ for some constant $C \leq 1$.
\end{restatable}

\subsection{Nesting HCS with Regular Guards} 
\label{subsec:nesting}

In principle, it is possible to use HCS as guards. So, for example, \automata{NFA}{\automata{NFA}{NFA}} is well-defined.  Let $\Nn_k$ denote the class of automata that are NFAs with nested NFA guards of depth $k$; $\Nn_1 =$ \automata{NFA}{NFA}, $\Nn_2 =$ \automata{NFA}{\automata{NFA}{NFA}}, and so on.  By induction we can see that $\Nn_i$ recognise regular languages. 
However, by carefully applying~\cref{thm:exponential_size} on the guard by induction, we can observe that the number of states in the DFA remains at most exponential. See the full proof in~\cref{appendix:succinctness}.

\begin{restatable}{proposition}{nestingnomore}
	Let $k \in \NN$ and let $A \in \Nn_k$ be an NFA with nested NFA guards that contains $n$ states in total (across the underlying NFA and all of its nested guard NFAs).
	There exists a DFA $D$ such that $\lang{D} = \lang{A}$ and the number of states in $D$ is at most $2^n$.
\end{restatable}
 
\section{Games on HCS are \class{EXPTIME}-complete}\label{sec:games}

We now consider games on history-constrained systems. 
Two-player zero-sum games played on graphs are common in formal verification,
and formalise system correctness even in antagonistic environments.
We refer to~\cite{GOG} for a comprehensive overview of games on graphs.

In this context, we will henceforth assume that the HCS under consideration is \emph{non-blocking}, meaning that for every history $h=q_0a_1q_1a_2\ldots q_k \in Q\x(\Sigma\x Q)^*$ there is at least one admissible continuation $h\, a_{k+1} q_{k+1}$ enabled via a transition that is from $q_k$, reads $a_{k+1} \in \Sigma$, leads to $q_{k+1}$, is guarded by $L$, and $a_1\ldots a_k\in L$. This is not a restriction because one can always make any HCS non-blocking by adding unguarded transitions to (winning or losing) sinks.

\begin{definition}[Games on HCS]
Let $A = (\Sigma, U, \Gg)$ be a non-blocking HCS and let $Q$ denote the set of control states of the underlying automaton $U$.

A game on $A$ is played by two opposing players, 0 and 1, and determined by
a predicate $\owner:Q\to\{0,1\}$, that associates with each state of $U$ which player owns it,
and a \emph{winning condition} $\WC \subseteq (Q\x\Sigma)^{\omega}$, that declares which outcomes are considered a win for player 0.

The game is played in rounds, each of which extends the current history play $h\in Q\x(\Sigma\x Q)^*$, starting at the initial state $q_0\in Q$.
From history \hbox{$h=q_0a_0\ldots q_k$}, player $\owner(q_k)$ selects an admissible move $q_k\step{a_k}q_{k+1}$ and the game continues with history
$h\,a_k q_{k+1}$. In the limit, this describes an infinite sequence $\rho\in(Q\x\Sigma)^{\omega}$ called a \emph{play}. 
The play is won by Player 0 if $\rho\in\WC$, otherwise the play is won by Player 1.
Player $i$ \emph{wins the game} if there is a strategy to select admissible moves for every history such that regardless of the opponent's choices, the resulting play is won by Player $i$.
\end{definition}

Solving a game refers to the decision question to verify whether a given game is won by Player~$0$.
The decidability, and complexity, of game solving crucially depend on how the game is specified. 

The main result of this section is the following theorem that concerns reachability games on finite-state HCS.
Here, HCS are of the form \automata{NFA}{DFA} and the winning condition $\WC$ is defined by a subset $F\subseteq Q$ of target states that Player~$0$ aims to reach (i.e., $h \in \WC$ if and only if $q_i \in F$ for some $0\le i$).

\begin{theorem}\label{thm:exptime-completeness}
	Solving reachability games on \automata{NFA}{DFA} is \class{EXPTIME}-complete.
\end{theorem}

Before the proof, let's justify why these types of games are the most interesting and appropriate to study here.

Firstly, we argue that the dynamics of the game based on ownership of states is natural for modelling purposes, where one verifies robustness of protocols/programs against antagonistic environments. Conceivable alternatives where players take turns, or where one player picks an action in $\Sigma$ and the other a such-labelled transition, can be easily encoded using the finite-state control.

Our choice to focus on \automata{NFA}{DFA}, in particular compared to more expressive guard automata, is natural because determinising NFA guards incurs negligible additional cost on top of having an NFA as the underlying automaton.
Furthermore, allowing guards defined in terms of even more expressive models, like Vector Addition Systems with States (VASS) or PDA, quickly renders solving any kind of game on them undecidable.
Indeed, even using guards defined as coverability languages for 1-VASS, a heavily restricted subclass of context-free languages, allows one to reduce the reachability problem for two-counter machines to solving games on 2-VASS~\cite{vass-games-undec} and completely analogously to games on HCS with VASS guards.

Finally, our subsequent focus on the reachability winning conditions is because more expressive alternatives commonly studied in formal verification / reactive synthesis can be solved using the same, na\"ive approach.

\subsection{Exponential Time Upper Bound}

To solve a game on \automata{NFA}{DFA}, we can employ a construction that is very similar to the construction presented in~\cref{thm:exponential_size}.
We construct a finite-state game graph where the vertices represent configurations of the underlying NFA and all guard DFAs by using the product construction.
Solving the original HCS game amounts to solving a game on the constructed graph, where a state's ownership, and other properties relevant to the winning condition like membership in $F$, are directly lifted to the game graph and its winning condition.
Crucially, the size of this graph is exponential in the number of guards only.
For~\cref{lem:game-graph}, suppose that $\delta = \set{t_1, \ldots, t_m}$ are the transitions of the underlying NFA and the DFA guarding $t_i$ has states $Q_i$ and transitions $\delta_i$.
The following lemma is proved in~\cref{appendix:games}.

\begin{restatable}{lemma}{gamegraph}\label{lem:game-graph}
	Given an HCS game $G$ one can in polynomial time construct a directed graph $G'=(V', E')$, a predicate $\owner': V' \to \set{0,1}$, and an initial state $q_0'\in V'$
	such that
	\begin{enumerate}[(i)]
		\item Player~0 wins $G$ iff Player~$0$ wins $G'$;
		\item $\abs{V} \leq \abs{Q} \cdot (\abs{Q_1} + \ldots + \abs{Q_m})^m$ and $\abs{E} \leq \abs{\delta} \cdot (\abs{\delta_1} + \ldots + \abs{\delta_m})^m$.
	\end{enumerate}
\end{restatable}	

Solving the resulting reachability games can be done in polynomial time~\cite[Theorem 1.1]{GOG} and therefore games on HCS can be solved in exponential time.
The same construction allows one to solve history-constrained games with more expressive winning conditions, as long as solving their ordinary finite-state variants only depends polynomially on the number of control states.
This applies to parity games (in $O(m\cdot n^k)$ \cite{zielonka1998infinite}) as well as energy games and mean-payoff games (in $O(m\cdot n\cdot W)$ \cite{energy-games}). Since each of these subsumes reachability, the lower bound (to be proven next) also applies to them, entailing:

\begin{corollary}
Solving parity, energy, and mean-payoff games on \automata{NFA}{DFA} is \class{EXPTIME}-complete.
\end{corollary}

\subsection{Exponential Time Lower Bound}
\label{sub:lower_bound_for_reachability_games}

To conclude the proof of \cref{thm:exptime-completeness} we show the lower bound by reduction from solving \emph{Countdown Games}, which is \class{EXPTIME}-complete
\cite[Theorem 4.5]{JurdzinskiSL08}.

A countdown game is played on a weighted directed graph, where each edge has a weight in $\NN\setminus \{0\}$, given in binary. A configuration is a pair ($s,x)$, where $s$ is a state and $x \in \NN$ is the counter value.
In any configuration $(s,x)$, Player~0 chooses a number $d$ such that there is at least one outgoing edge from $s$ with weight $d$. Player 1 then chooses any edge $(s, d, s')$ that starts in $s$ and has weight $d$, and the game continues in configuration $(s',x-d)$. The goal of Player 0 is to make the counter reach 0, the goal of Player 1 is to prevent this. A countdown-game is therefore a concisely represented reachability game on $S\x\N$.

The above is the usual representation of countdown games. Here it is convenient to represent them slightly differently, however. Firstly, instead of counting down from $x$ to 0, we will count up from $0$ to $x$. Secondly, we ``unravel'' the choices of Players 0 and 1, meaning that we first have Player 0 choose an auxiliary successor state $s_d$ by taking an edge $(s, d, s_d)$, and then Player 1 choose a successor $s'$ of $s_d$ by taking an edge $(s_d, 0, s')$, %
see Figure~\ref{fig:adder}.
Furthermore, we assume without loss of generality that all edge weights as well as the target number $x$ are a power of 2 (or 0). %

\begin{lemma}\label{lem:exptime-reduction}
	Let $\Cc$ be a countdown game. We can construct, in polynomial time, an NFA[DFA] reachability game $\Rr$ such that Player~0 has a winning strategy in $\Rr$ if and only if Player~0 has a winning strategy in $\Cc$.
\end{lemma}

\begin{proof}
	The main idea is to use the guarding DFAs to maintain the current counter value. 
	Let $2^k$ be the target value and $x \leq 2^k$ denote the counter.
	We will have $k+1$ many DFA guards $D_0, D_1, \ldots, D_k$, which are shown in Figure~\ref{fig:binary-counter}.
	The role of the $i$-th DFA guard $D_i$ is to keep track of the $i$-th bit of $x$: state $p_i$ represents 0, while $q_i$ represents 1. The state $a_i$ is used to force a \emph{carry} operation when necessary, and $r_i$ is a rejecting sink state.
	
		\begin{figure}[ht]
			\centering
\begin{tikzpicture}
	\node[state, inner sep = 0.5mm, Red] (p0) at (0,0.5) {\small$p_0$};
	\node[state, inner sep = 2.3mm, Red] (p0) at (0,0.5) {};
	\node[state, inner sep = 1mm, Red] (q0) at (3,0.5) {\small$q_0$};
	\node[state, inner sep = 1mm, Red] (a0) at (1.5,-0.5) {\small$a_0$};
	\node[state, inner sep = 1mm, Gray] (r0) at (2.5,-2) {\small$r_0$};
	\draw[transition, Red] (p0) edge node[above,scale=0.8] {$2^0$} (q0);
	\draw[transition, Red] (q0) edge node[below,scale=0.8] {$2^0$} (a0);
	\draw[transition, Red] (a0) edge node[below,scale=0.8] {$\texttt{c}_0$} (p0);
	\draw[transition, Gray] (p0) edge[out = 250, in = 200, distance = 40] node[left,scale=0.8] {$\texttt{c}_0$} (r0);
	\draw[transition, Gray] (a0) edge[bend left=20] node[right=0.1, pos=0.2, scale=0.8] {$2^0, \texttt{n}_0$} (r0);
	\draw[transition, Gray] (q0) edge[bend left=20] node[right,scale=0.8] {$\texttt{c}_0$} (r0);
	\draw[transition, Red] (p0) edge[loop above, out=70, in=110, distance = 6mm] node[above,scale=0.8] {$\Sigma\setminus\set{2^0, \texttt{c}_0}$} (p0);
	\draw[transition, Red] (q0) edge[loop above, out=70, in=110, distance = 6mm] node[above,scale=0.8] {$\Sigma\setminus\set{2^0, \texttt{c}_0}$} (q0);
	\draw[transition, Red] (a0) edge[loop above, out=210, in=250, distance = 6mm] node[below, pos=0.6, scale=0.8] {$\hspace{0.1cm}\Sigma\setminus\set{2^0, \texttt{c}_0, \texttt{n}_0}$} (a0);
	\draw[transition, Gray] (r0) edge[loop above, out=330, in=290, distance = 6mm] node[pos=0.3,right,scale=0.8] {$\Sigma$} (r0);

	\node[state, inner sep = 0.5mm, Green] (pi) at (5,0.5) {\small$p_i$};
	\node[state, inner sep = 2.3mm, Green] (pi) at (5,0.5) {};
	\node[state, inner sep = 1mm, Green] (qi) at (8,0.5) {\small$q_i$};
	\node[state, inner sep = 1mm, Green] (ai) at (6.5,-0.5) {\small$a_i$};
	\node[state, inner sep = 1mm, Gray] (ri) at (7.5,-2) {\small$r_i$};
	\draw[transition, Green] (pi) edge node[above,scale=0.8] {$2^i, \texttt{c}_{i-1}$} (qi);
	\draw[transition, Green] (qi) edge node[below=0.2, pos = 0.25, scale=0.8] {$2^i, \texttt{c}_{i-1}$} (ai);
	\draw[transition, Green] (ai) edge node[below, scale=0.8] {$\texttt{c}_i$} (pi);
	\draw[transition, Green] (pi) edge[loop above, out=70, in=110, distance = 6mm] node[above,scale=0.8] {$\Sigma\setminus\set{2^i, \texttt{c}_i, \texttt{c}_{i-1}}$} (pi);
	\draw[transition, Green] (qi) edge[loop above, out=70, in=110, distance = 6mm] node[above,scale=0.8] {$\Sigma\setminus\set{2^i, \texttt{c}_i, \texttt{c}_{i-1}}$} (qi);
	\draw[transition, Gray] (pi) edge[out = 250, in = 200, distance = 40] node[left,scale=0.8] {$\texttt{c}_i$} (ri);
	\draw[transition, Gray] (ai) edge node[right=0.03, pos=0.1, scale=0.8] {$2^i, \texttt{c}_{i-1}, \texttt{n}_i$} (ri);
	\draw[transition, Gray] (qi) edge[bend left = 34] node[right,scale=0.8] {$\texttt{c}_i$} (ri);
	\draw[transition, Gray] (ri) edge[loop above, out=330, in=290, distance = 6mm] node[pos=0.3,right,scale=0.8] {$\Sigma$} (ri);
	\draw[transition, Green] (ai) edge[loop above, out=210, in=250, distance = 6mm] node[below=0.05, scale=0.8] {$\Sigma\setminus\{2^i, \texttt{c}_i,\hspace{0.2cm}$} node[below=0.35, scale=0.8] {$\hspace{1cm}\texttt{c}_{i-1}, \texttt{n}_i\}$} (ai);

	\node[state, inner sep = 1mm, Blue] (pk) at (10.75,0.5) {\small$p_k$};
	\node[state, inner sep = 0.5mm, Blue] (qk) at (10.75,-1) {\small$q_k$};
	\node[state, inner sep = 2.3mm, Blue] (qk) at (10.75,-1) {};
	\node[state, inner sep = 1mm, Gray] (rk) at (10.75,-2.5) {\small$r_k$};
	\draw[transition, Blue] (pk) edge[loop above, out=70, in=110, distance = 6mm] node[above,scale=0.8] {$\Sigma\setminus\set{2^k, \texttt{c}_{k-1}}$} (pk);
	\draw[transition, Blue] (pk) edge node[left,pos=0.35,scale=0.8] {$2^k, \texttt{c}_{k-1}$} (qk);
	\draw[transition, Blue] (qk) edge[loop above, out=200, in=160, distance = 6mm] node[above,pos=0.6,scale=0.8] {$\texttt{n}_0, \texttt{n}_1,\hspace{1.15cm}$} node[left,scale=0.8] {$\ldots,$} node[below=-0.04,pos=0.4,scale=0.8] {$\texttt{n}_{k-1}\hspace{1cm}$}  (qk);
	\draw[transition, Gray] (qk) edge node[left,pos=0.3,scale=0.8] {$\Sigma\setminus\{\texttt{n}_0, \texttt{n}_1,\hspace{0.3cm}$}  node[left,pos=0.6,scale=0.8] {$\ldots, \texttt{n}_k\}$} (rk);
	\draw[transition, Gray] (rk) edge[loop above, out=200, in=160, distance = 6mm] node[left,scale=0.8] {$\Sigma$} (rk);
\end{tikzpicture}

 			\caption{
				The DFA guards; $D_0$ (left, \textcolor{Red}{red}), $D_i$ for all $1 \leq i \leq k-1$ (centre, \textcolor{Green}{green}), and $D_k$ (right, \textcolor{Blue}{blue}). 
				Note that the sink states $r_0, \ldots, r_k$ and all transitions to/from sink states have been coloured grey to highlight the main function of the DFAs which occurs between states $p_i$, $q_i$, and $a_i$.
			}
			\label{fig:binary-counter}
		\end{figure}
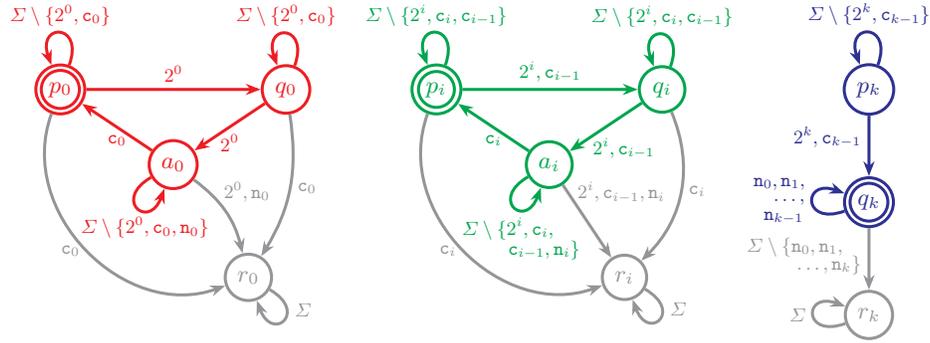

	The counter starts at $x=0$, so the initial state of $D_i$ is $p_i$, for $0 \leq i \leq k$.
	As $2^k$ is the target, the accepting states are $p_i$ in $D_i$, for $0 \leq i < k$, and $q_k$ in $D_k$.

	Now, let us consider the game NFA, i.e., the outer automaton (see~\cref{fig:adder}). 
	Suppose that in the ``countup'' game, Player 0 chooses value $2^i$. We represent this using $k+1$ transitions in our NFA. The first transition is labelled $2^i$. This symbol is also read by the guard automaton $D_i$, which increments from value 0 (state $p_i$) to value 1 ($q_i$), or from value 1 ($q_i$) to ``preparing to carry'' ($a_i$). 
	All other guards $D_j$ (for $j\neq i$) ignore the symbol $2^i$, i.e., they loop and remain in the same state. 
	Then, Player~0 chooses either $\texttt{c}_0$ or $\texttt{n}_0$, followed by $\texttt{c}_1$ or $\texttt{n}_1$, up to $\texttt{c}_{k-1}$ or $\texttt{n}_{k-1}$. 
	After the final choice (between $\texttt{c}_{k-1}$ or $\texttt{n}_{k-1}$), we arrive in a state where Player 1 can choose a successor of $s_0$ among the ones where an edge with value $2^i$ exists in the countup game. 

	\vspace{-2mm}
	\begin{figure}
		\centering
\begin{tikzpicture}
	\node[player0] (s0) at (0,0) {\small$s_0$};
	\node[player1] (t1) at (1,0.75) {};
	\node[player1] (t2) at (1,-0.75) {};
	\node[player0] (s1) at (2,1.5) {\small$s_1$};
	\node[player0] (s2) at (2,0.75) {\small$s_2$};
	\node[player0] (s3) at (2,0) {\small$s_3$};
	\node[player0] (s4) at (2,-0.75) {\small$s_4$};
	\draw[transition] (s0) -- node[above=0.1, pos=0.2, scale=0.8] {$+2^i$} (t1);
	\draw[transition] (s0) -- node[below=0.1, pos=0.2, scale=0.8] {$+2^j$} (t2);
	\draw[transition] (t1.center) -- (s1);
	\draw[transition] (t1.center) -- (s2);
	\draw[transition] (t1.center) -- (s3);
	\draw[transition] (t2.center) -- (s4);
	
	\node[player0] (s0new) at (3.5, 0) {\small$s_0$};
	\node[player0] (ci1) at (4.5, 0.75) {};
	\node[player0] (cj1) at (4.5, -0.75) {};
	\node[player0] (ci2) at (5.75, 0.75) {};
	\node[player0] (cj2) at (5.75, -0.75) {};
	\node[regular polygon,regular polygon sides=4, scale=0.8, fill = white, inner sep = 0] (ci3) at (7, 0.75) {\small$\cdots$};
	\node[regular polygon,regular polygon sides=4, scale=0.8, fill = white, inner sep = 0] (cj3) at (7, -0.75) {\small$\cdots$};
	\node[player1] (ti) at (8.25, 0.75) {};
	\node[player1] (tj) at (8.25, -0.75) {};
	\node[player0] (s1new) at (9.5,1.5) {\small$s_1$};
	\node[player0] (s2new) at (9.5,0.75) {\small$s_2$};
	\node[player0] (s3new) at (9.5,0) {\small$s_3$};
	\node[player0] (s4new) at (9.5,-0.75) {\small$s_4$};
	\draw[transition] (s0new) -- node[above=0.1, pos=0.3, scale=0.8] {$2^i$} (ci1);
	\draw[transition] (s0new) -- node[below=0.1, pos=0.3, scale=0.8] {$2^j$} (cj1);
	\draw[transition] (ci1) edge[out = 20, in = 160] node[above, scale=0.8] {$\texttt{c}_0$} (ci2);
	\draw[transition] (ci1) edge[out=-20, in = -160] node[below, scale=0.8] {$\texttt{n}_0$} (ci2);
	\draw[transition] (cj1) edge[out = 20, in = 160] node[above,  scale=0.8] {$\texttt{c}_0$} (cj2);
	\draw[transition] (cj1) edge[out=-20, in = -160] node[below,  scale=0.8] {$\texttt{n}_0$} (cj2);

	\draw[transition] (ci2) edge[out = 20, in = 160] node[above, scale=0.8] {$\texttt{c}_1$} (ci3);
	\draw[transition] (ci2) edge[out=-20, in = -160] node[below, scale=0.8] {$\texttt{n}_1$} (ci3);
	\draw[transition] (cj2) edge[out = 20, in = 160] node[above, scale=0.8] {$\texttt{c}_1$} (cj3);
	\draw[transition] (cj2) edge[out=-20, in = -160] node[below, scale=0.8] {$\texttt{n}_1$} (cj3);

	\draw[transition] (ci3.center) edge[out = 20, in = 160] node[above, pos=0.6, scale=0.8] {$\texttt{c}_{k-1}$} (ti);
	\draw[transition] (ci3.center) edge[out=-20, in = -160] node[below, pos=0.6, scale=0.8] {$\texttt{n}_{k-1}$} (ti);
	\draw[transition] (cj3.center) edge[out = 20, in = 160] node[above, pos=0.6, scale=0.8] {$\texttt{c}_{k-1}$} (tj);
	\draw[transition] (cj3.center) edge[out=-20, in = -160] node[below, pos=0.6, scale=0.8] {$\texttt{n}_{k-1}$} (tj);

	\draw[transition] (ti.center) -- node[above,pos=0.6,scale=0.8]{$\varepsilon$} (s1new);
	\draw[transition] (ti.center) -- node[above,pos=0.7,scale=0.8]{$\varepsilon$} (s2new);
	\draw[transition] (ti.center) -- node[below,pos=0.6,scale=0.8]{$\varepsilon$} (s3new);
	\draw[transition] (tj) -- node[above,pos=0.45, scale=0.8]{$\varepsilon$} (s4new);

	\node[player0] (end1) at (4.5, -1.75) {};
	\node[player0] (end2) at (5.75, -1.75) {};
	\node[regular polygon,regular polygon sides=4, scale=0.8, fill = white, inner sep = 0] (end3) at (7, -1.75) {$\small\cdots$};
	\node[player0] (end4) at (8.25, -1.75) {};
	\node[player0,accepting] (end5) at (9.5, -1.75) {};

	\draw[transition] (s0new) edge[bend right=30] node[left=0.1,pos=0.66, scale=0.8, black]{$\varepsilon$} (end1);
	\draw[transition] (end1) edge[Red] node[above,black, scale=0.8]{$\varepsilon$} node[below, scale=0.8]{$D_0$} (end2);
	\draw[transition] (end2) edge[Green] node[above,black, scale=0.8]{$\varepsilon$} node[below, scale=0.8]{$D_1$} (end3);
	\draw[transition] (end3) edge[Green] node[above,black, scale=0.8]{$\varepsilon$} node[below, scale=0.8]{$D_{k-1}$} (end4);
	\draw[transition] (end4) edge[Blue] node[above,black, scale=0.8]{$\varepsilon$} node[below, scale=0.8]{$D_k$} (end5);

	\node[player1] (t1) at (1,0.75) {};
	\node[player1] (t2) at (1,-0.75) {};

	\node[player1] (ti) at (8.25, 0.75) {};
	\node[regular polygon,regular polygon sides=4, scale=0.8, fill = white, inner sep = 0] (ci3) at (7, 0.75) {\small$\cdots$};
	\node[regular polygon,regular polygon sides=4, scale=0.8, fill = white, inner sep = 0] (cj3) at (7, -0.75) {\small$\cdots$};
\end{tikzpicture} 		\caption{
			Part of an example ``countup'' game with five circle states controlled by Player~0 and two square states controlled by Player~1 (left) and its corresponding part in the \automata{NFA}{DFA} game (right). Outgoing edges of $s_1,s_2,s_3,s_4$ are omitted.
		}
		\label{fig:adder}
	\end{figure}
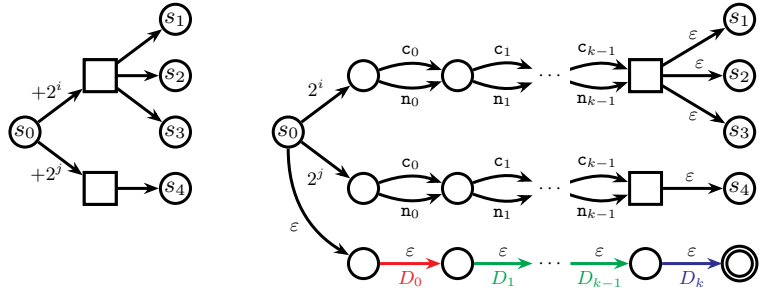
	
	Additionally, in every state $s$, Player~0 may also choose to take an $\varepsilon$-transition that leads to a chain of $\varepsilon$-transitions guarded by $D_0, \ldots, D_k$ and terminating in (the only) accepting state. 
	The eventual goal of Player 0 is to reach this accepting state, which requires all $D_i$, $0\leq i \leq k$, to be in an accepting state. 
	See Figure~\ref{fig:adder}.
	
	When Player~0 chooses the action $\texttt{c}_i$, this represents the $i$-th bit carrying over to the $(i+1)$-st bit. Hence the automaton $D_i$ resets from state $a_i$ (``preparing to carry'') to $p_i$ (value 0), while the automaton $D_{i+1}$ increments from $p_{i+1}$ to $q_{i+1}$ or from $q_{i+1}$ to $a_{i+1}$. If $\texttt{c}_i$ is invoked when $D_i$ is not in the carry state $a_i$, then the automaton moves into the sink state $r_i$. In order to reach the accepting state in the outer automaton, all $D_i$ must be in an accepting state, so reaching $r_i$ is immediately losing for Player~0. 
	Hence if $D_i$ is not preparing to carry, then playing $\texttt{c}_{i}$ is a losing move.
	The action $\texttt{n}_i$, conversely, represents no carry being required. This moves $D_i$ to the sink state $r_i$ if it is in the carry state $a_i$, so Player~0 must similarly not play $\texttt{n}_i$ if $D_i$ is in $a_i$.
	
	Hence the only way for Player~0 to win is to correctly play $\texttt{c}_i$ whenever the $i$-th bit needs to be carried and $\texttt{n}_i$ otherwise, which means that $D_0, D_1, \ldots, D_k$ accurately track the value of the counter $x$. Then, once $x=2^k$, all $D_i$ will be in their accepting state, so Player~0 can take the chain of $\varepsilon$-transitions, guarded by $D_0, D_1, \ldots, D_k$ to finally reach the accepting state and win the game.

	This \automata{NFA}{DFA} therefore satisfies the requirements of the Lemma.\qed

\end{proof}

\section{HCS with VASS Guards}\label{sec:vass}

As a language recognition model, a $d$-dimensional VASS $V = (\Sigma, Q, T, q_0, F)$ consists of a finite alphabet $\Sigma$, a finite set of states $Q$, a finite set of transitions $T \sset Q \times \Sigma \times \ZZ^d \times Q$, an initial state $q_0 \in Q$, and a set of accepting states $F \sset Q$.
A configuration of $V$ is a pair $(q, \vec{x})$ where $q \in Q$ is the current state and $\vec{x} \in \NN^d$ are the current counter values.
Notice that the counter values are non-negative. 
To be clear, a transition $(q, \sigma, \vec{u}, q') \in T$ can be taken from $(q, \vec{x})$ if $\vec{x} + \vec{u} \geq \vec{0}$ (this inequality holds on all counters). A VASS is \emph{deterministic}, denoted DVASS, if there is at most one transition for each outgoing state $q \in Q$ and character $\sigma \in \Sigma$.

The language defined by a given VASS depends on its acceptance condition.
There are two acceptance conditions: reachability and coverability.
The reachability acceptance condition necessitates that a word is accepted if it leads to a configuration $(q, \vec{x})$ such that $q \in F$ is an accepting state and $\vec{x} = \vec{0}$.
The coverability acceptance condition is a relaxation that only requires that $q \in F$ and $\vec{x} \geq \vec{0}$ (i.e.\ only the control state matters).
We use \emph{Reach-VASS} to denote VASS with the reachability acceptance condition and \emph{Cover-VASS} to denote VASS with the coverability acceptance condition.
When the dimension $d$ of the VASS is known, we write \emph{Reach-$d$-VASS} and \emph{Cover-$d$-VASS}.%

This section is devoted to proving that language non-emptiness is undecidable for \automata{DFA}{Reach-1-DVASS} (\cref{thm:reach-vass-undecidable}) and is \class{EXPSPACE}-complete for \automata{NFA}{Cover-DVASS} (\cref{thm:cover-vass-decidable}).

\subsection{VASS Guards with Reachability Acceptance Conditions}

\begin{restatable}{theorem}{reachvassundecidable}\label{thm:reach-vass-undecidable}
	Language emptiness for \automata{DFA}{Reach-1-DVASS} is undecidable.
\end{restatable}

To prove~\cref{thm:reach-vass-undecidable}, we will reduce from reachability in two-counter machines (2CM), an undecidable problem~\cite{minsky1967computation}.
We will construct an \automata{DFA}{Reach-1-DVASS} that has a non-empty language if and only if reachability holds in a given 2CM.

A 2CM $M = (Q, T)$ consists of a finite set of states $Q$ and a set of transitions $T \sset Q \times A \times Q$ where $A = \set{\texttt{inc}_1, \texttt{dec}_1, \texttt{zero}_1, \texttt{inc}_2, \texttt{dec}_2, \texttt{zero}_2}$ is the set of \emph{counter actions}.
A configuration of $M$ is a pair $(q, (x, y))$ where $q \in Q$ is the current state and $(x, y) \geq (0, 0)$ are the current counter values. 
To be clear, a transition $(p, \texttt{dec}_1, q)$ can be taken from $(p, (x,y))$ to $(q, (x-1, y))$ if $x-1 \geq 0$ and a transition $(p, \texttt{zero}_2, q)$ can be taken from $(p, (x, y))$ to $(q, (x, y))$ if $y = 0$.
We may also assume ``action determinacy'': for every state $p \in Q$ and for every action $\alpha \in A$, there is at most one transition from $p$ of the form $(p, \alpha, q)$.

\begin{proofidea}[Full proof in~\cref{appendix:reach-vass}]We create a \automata{DFA}{Reach-1-DVASS} $D$, over the alphabet $\set{\texttt{inc}_1, \texttt{dec}_1, \texttt{zero}_1, \texttt{inc}_2, \texttt{dec}_2, \texttt{zero}_2}$, where the underlying DFA $U$ tracks the current state of the given 2CM $M$.
There are two \emph{distinct} Reach-1-DVASS guards $G_1$ and $G_2$, each consisting of only one state.
The roles of $G_1$ and $G_2$ are to track the first  and second counter of $M$, respectively.
Whenever there is a transition $(p, \texttt{zero}_i, q)$, it is guarded by $G_i$; this equates to the $i$-th counter of $M$ being zero. 
See~\cref{fig:reach-vass-example} for an example of this construction.\qed

\begin{figure}
	\centering
\begin{tikzpicture}
	\node[state] (p1) at (0,0) {};
	\node[state] (p2) at (1.8,0.5) {};
	\node[state] (p3) at (1.8,-0.5) {};
	\node[state] (p4) at (1.1,1.5) {};
	\node[state] (p5) at (2.5,1.5) {};
	\node[state] (p6) at (3.6,0) {};
w
	\draw[transition] (p1) edge[loop above, out=110, in=70, distance = 6mm] node[above] {$\texttt{inc}_1$} (p1);
	\draw[transition] (p1) -- node[above, sloped]{$\texttt{zero}_2$} (p2);
	\draw[transition] (p1) -- node[below, sloped]{$\texttt{zero}_1$} (p3);
	\draw[transition] (p3) -- node[right]{$\texttt{inc}_2$} (p2);
	\draw[transition] (p2) -- node[left]{$\texttt{dec}_1$} (p4);
	\draw[transition] (p4) -- node[above]{$\texttt{inc}_2$} (p5);
	\draw[transition] (p5) -- node[right]{$\texttt{inc}_2$} (p2);
	\draw[transition] (p2) -- node[above, sloped]{$\texttt{zero}_1$} (p6);

	\node[state] (q1) at (4.7,0) {};
	\node[state] (q2) at (6.5,0.5) {};
	\node[state] (q3) at (6.5,-0.5) {};
	\node[state] (q4) at (5.8,1.5) {};
	\node[state] (q5) at (7.2,1.5) {};
	\node[state] (q6) at (8.3,0) {};

	\draw[transition] (q1) edge[loop above, out=110, in=70, distance = 6mm] node[above] {$\texttt{inc}_1$} (q1);
	\draw[transition, Green] (q1) -- node[above, sloped]{$\texttt{zero}_2$} node[below=-0.02, sloped, pos=0.75, scale=0.7]{\small $G_2$}(q2);
	\draw[transition, Red] (q1) -- node[below, sloped]{$\texttt{zero}_1$} node[above=-0.02, sloped, pos=0.75, scale=0.7]{$G_1$} (q3);
	\draw[transition] (q3) -- node[right]{$\texttt{inc}_2$} (q2);
	\draw[transition] (q2) -- node[left]{$\texttt{dec}_1$} (q4);
	\draw[transition] (q4) -- node[above]{$\texttt{inc}_2$} (q5);
	\draw[transition] (q5) -- node[right]{$\texttt{inc}_2$} (q2);
	\draw[transition, Red] (q2) -- node[above, sloped]{$\texttt{zero}_1$} node[below=-0.02, sloped, pos=0.75, scale=0.7]{$G_1$} (q6);

	\draw[rounded corners, transition, Red] (9.4,2.1) rectangle (11.4,0.7);
	\node[Red, scale=1.2] at (9.1, 1.9) {$G_1$};
	\node[state, Red,accepting] (g1) at (9.8, 1.4) {};
	\draw[transition, Red] (g1) edge[loop above, out=110, in=70, distance = 5mm] node[right, pos=0.8] {$\texttt{inc}_1, +1$} (g1);
	\draw[transition, Red] (g1) edge[loop below, out=-110, in=-70, distance = 5mm] node[right, pos=0.8] {$\texttt{dec}_1, -1$} (g1);

	\draw[rounded corners, transition, Green] (9.4,0.5) rectangle (11.4,-0.9);
	\node[Green, scale=1.2] at (9.1, 0.3) {$G_2$};
	\node[state, Green,accepting] (g2) at (9.8, -0.2) {};
	\draw[transition, Green] (g2) edge[loop above, out=110, in=70, distance = 5mm] node[right, pos=0.8] {$\texttt{inc}_2, +1$} (g2);
	\draw[transition, Green] (g2) edge[loop below, out=-110, in=-70, distance = 5mm] node[right, pos=0.8] {$\texttt{dec}_2, -1$} (g2);
	
\end{tikzpicture} 	\caption{A 2CM and its equivalent \automata{DFA}{Reach-1-DVASS}.}
	\label{fig:reach-vass-example}
\end{figure}
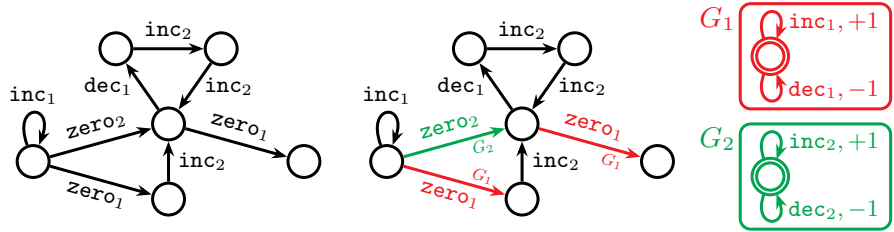
\end{proofidea}

A key reason that non-emptiness in \automata{DFA}{Reach-1-DVASS} is undecidable is that we can use the guards to simulate zero tests as many times as we like.
When the number of zero tests of a 2-CM $M$ is fixed, or even when it is not fixed but is bounded by some function of size of $M$, then reachability becomes decidable.
This is because VASS can simulate counter machines with a bounded number of zero tests and reachability in VASS is decidable~\cite{Mayr81,Kosaraju82,Lambert92,LerouxS15,LerouxS19}.

\subsection{VASS Guards with Coverability Acceptance Conditions}
\label{subsec:cover-vass}

This section observes that \automata{NFA}{Cover-DVASS} are contained within the class of Cover-VASS (\cref{pro:existsequivalentvass}), but despite their succinct representation, the language emptiness problem is \class{EXPSPACE}-complete, matching the complexity for Cover-DVASS (\cref{thm:cover-vass-decidable}). 

\begin{theorem}\label{thm:cover-vass-decidable}
	Language emptiness for \automata{NFA}{Cover-DVASS} is \class{EXPSPACE}-complete.
\end{theorem}

Let us fix our attention on the \automata{NFA}{Cover-DVASS} $A = (\Sigma, U, \Gg)$ where $U = (\Sigma, Q, \delta, q_0, F)$ is the underlying NFA with transitions $\delta = \set{t_1, \ldots, t_m}$ and guards $\Gg = \set{ L_1, \ldots, L_m}$.
For every $i \in \set{1, \ldots, m}$, the guard is specified as a $d_i$-dimensional Cover-DVASS $G_i = (\Sigma, Q_i, T_i, q_{0,i}, F_i)$, so $\lang{G_i} = L_i$.
We use $\norm{V}$ to denote the size of a $d$-VASS encoded in unary.

\begin{restatable}{proposition}{existsequivalentvass}\label{pro:existsequivalentvass}
	Let $D = d_1 + \ldots + d_m$.
	There exists a $D$-dimensional Cover-VASS $V = (\Sigma, Q', T', q_0',F' = \set{t})$ such that $L(A) = L(V)$ and $\norm{V} \leq 2(\norm{U} + \norm{G_1} + \ldots + \norm{G_m})^{m+2}$.
\end{restatable}

\begin{proofidea}[Full proof in \cref{appendix:cover_vass}] The construction takes the cross product of the underlying automaton $U$ and the Cover-DVASS guards $G_1,\dots,G_m$ simultaneously tracking the counters of all of the guards. 
We must be careful to show that a transition that would kill a run of a guard does not kill the whole run (the guard may not be used again). 
We also ensure that there is a single accepting state so we can apply arguments and algorithms for VASS coverability.\qed
\end{proofidea}

\cref{pro:existsequivalentvass} tells us that \automata{NFA}{Cover-DVASS} are no more expressive than Cover-VASS. 
It remains to show that deciding non-emptiness (whether the state $t$ is reachable in $V$) can be decided in \class{EXPSPACE} and that there is a matching lower bound.

\paragraph{Proof of \cref{thm:cover-vass-decidable}.}
For this, we will use an exponential space algorithm for deciding coverability.
Coverability in VASS is the following decision problem.
Given a VASS $V$, an initial configuration $(s, \vec{x})$, and a target configuration $(t, \vec{y})$, decide whether there is a run in $V$ from $(s, \vec{x})$ to $(t, \vec{y'})$ for some $\vec{y'} \geq \vec{y}$.
In 1978, Rackoff proved that coverability in VASS can be decided in \class{EXPSPACE}~\cite{Rackoff78}.
We will use a more convenient formulation of Rackoff's Theorem from~{\cite[Corollary 3.4]{KunnemannMSSW25}} that states that coverability in a $d$-VASS $V$ from an initial configuration $(s, \vec{0})$ and with a target configuration $(t, \vec{0})$ can be decided by a non-deterministic $\Oh(d^2\cdot2^d\cdot\log\norm{V})$-space algorithm.

For our purposes, we wish to decide coverability in the $D$-dimensional VASS $V = (\Sigma, Q', T', q_0, \set{t})$.
If we prove that $\norm{V}$ is at most exponential in the combined sizes of the underlying NFA and the Cover-DVASS guards, then we can use Rackoff's Theorem to conclude that language non-emptiness of $A$ is decidable in \class{NEXPSPACE}.
This is why, as per the statement of~\cref{thm:cover-vass-decidable}, we construct a VASS $V$ with $\norm{V} \leq 2(\norm{U} + \norm{G_1} + \ldots + \norm{G_m})^{m+2}$.
See~\cref{appendix:cover_vass}, specifically~\cref{clm:cover-vass-size}, for more details.

We conclude the upper bound by combining~\cref{pro:existsequivalentvass} and~{\cite[Corollary 3.4]{KunnemannMSSW25}}. 
There is a \class{NEXPSPACE} algorithm for deciding whether $(s, \vec{0})$ can cover $(t, \vec{0})$ in $V$.
Precisely, the algorithm uses at most the following space:
\begin{equation*}
	\Oh( d^2 \cdot 2^D \cdot (m+2) \cdot \log(2(\norm{U} + \norm{G_1} + \ldots + \norm{G_m})) ).
\end{equation*}
Finally, by Savitch's Theorem~\cite{Savitch70}, we know that \class{NEXPSPACE} $=$ \class{EXPSPACE}.

As for the lower bound, language non-emptiness of \automata{DFA}{Cover-DVASS} is \linebreak \class{EXPSPACE}-hard because coverability in VASS is already \class{EXPSPACE}-hard~\cite{Lipton76}. Associate with each transition an alphabet character, so that the Cover-DVASS is deterministic, generating an alphabet $\Sigma$.
Consider an DFA with two states $p$ (the initial state) and $q$ (the only accepting state) and transitions $t_1 = (p, \sigma, p)$ for each $\sigma\in\Sigma$, and $t_2 = (p, \$, q)$ for a new symbol $\$\not\in\Sigma$.
There is just one guard, the Cover-DVASS in question which is allowed to update its configuration whenever $t_1$ is taken and which guards the transition $t_2$.
This means that a word $w\$$ is accepted by the \automata{DFA}{Cover-DVASS} if and only if the guard is able to achieve coverability on the word $w$.
\qed

\subsubsection{Non-deterministic Cover-VASS Guards}
\label{subsec:non-deterministic-cover-vass}
We now consider the case non-determinstic Cover-VASS guards. Here we show that in general these are more powerful than just Cover-VASS by themselves. In particular, we leave open the question of whether emptiness is decidable for \automata{DFA}{Cover-VASS}.

\begin{restatable}{theorem}{vassnotcontained}
\label{thm:vassnotcontained}
The class \automata{DFA}{Cover-VASS} is not contained in Cover-VASS.
\end{restatable}

Our proof demonstrates a language which are recognised by \automata{DFA}{Cover-VASS} but not by Cover-VASS. In particular, we consider a variant of the Kleene star.

Recall that Kleene star of a language $L \sset \Sigma^*$ is $L^{*}= \set{\varepsilon} \cup L\cup LL\cup LLL \cup \cdots$.
We define the delimited Kleene star of $L$ to be $L^{\$*}= \$ \cup \$L\$ \cup \$L\$L\$ \cup \cdots$, where $\$$ is a symbol not in the alphabet of $L$, that is, we place a delimiter between each word in $L$.

To prove  \cref{thm:vassnotcontained} we consider the Cover-VASS language $L =\{a^n b^m : 0 \leq m \leq n\}$. 
In~\cref{claim:delkleenestar}, we prove that \automata{DFA}{Cover-VASS} can recognise the delimited Kleene star of any Cover-VASS language.
In particular, there exists a \automata{DFA}{Cover-VASS} that recognises $L^{\$*}$ (the delimited Kleene star of $L$).
However, in \cref{thm:cover-vass-kleene} we show that it cannot be recognised by any Cover-VASS.

In fact, as part of our proof we show that Cover-VASS are not closed under Kleene star, nor delimited Kleene star. 
Although the fact that Reach-VASS languages are not closed under Kleene star has been claimed by Hack in 1976~{\cite[Section 9.3.2]{Hack76}}, we are not aware of a direct proof in the literature for Cover-VASS. 
We prove it (\cref{thm:cover-vass-kleene}) using a similar using a proof strategy that is similar in nature to the proof of~{\cite[Lemma 3]{BosePT23}}.

\begin{restatable}{theorem}{covervasskleene}
\label{thm:cover-vass-kleene}
The class of languages accepted by Vector Addition Systems with States under coverability semantics (Cover-VASS) is not closed under Kleene star and not closed under delimited Kleene star.
\end{restatable}

\noindent The key element to proving \cref{thm:cover-vass-kleene} (full proof in \cref{appendix:non-det-cover-vass}) is the following claim:

\begin{restatable}{numberedclaim}{claimLstarnotcovervass}\label{claim:Lstar-not-cover-vass}
The language $L^* = \{a^n b^m : 0 \leq m \leq n\}^*$ is not accepted by any Cover-VASS.
\end{restatable}

Observe, being closed under delimited Kleene star would entail closure under Kleene star, hence $L^{\$*}$ can also not recognised by any Cover-VASS.

\begin{lemma}\label{claim:delkleenestar}
Given a Cover-VASS language $L$, \automata{DFA}{Cover-VASS} recognise the delimited Kleene star language $L^{\$*}$.
\end{lemma}
\begin{proof}
Let $L$ over $\Sigma$ be recognised by the Cover-VASS $V$. We define a \automata{DFA}{Cover-VASS} $A = (\Sigma\cup\{\$\}, U, G)$ which recognises  $L^{\$*}$. The set of guards contains just one guard language $G = \{G_\$\}$. 
We define $U$ and $G_\$$, as depicted in \cref{fig:kleene-star}.

Let $U$ have three states $q_0$, $q_1$, and $q_2$, with $q_0$ initial and $q_1$ accepting. 
The transitions of $U$ are:
\vspace{-\topsep}\begin{itemize}
	\item $q_0 \xrightarrow{\$} q_1$;
	\item  $q_1 \xrightarrow{\sigma} q_2$ and $q_2 \xrightarrow{\sigma} q_2$ for every $\sigma \in \Sigma$, and 
	\item $q_2 \xrightarrow{\$: G_\$} q_1$.
\end{itemize}

\begin{figure*}[t!]
    \centering
    \begin{subfigure}[t]{0.49\textwidth}
	    \centering\resizebox{0.99\textwidth}{!}{\begin{tikzpicture}[>=stealth, node distance=3cm, semithick]
		    \node[state] (q0) at (0,0) {$q_0$};
		    \node[state, accepting] (q1) at (2,0) {$q_1$};
		    \node[state] (q2) at (5,0) {$q_2$};
		    \draw[transition] (-0.7,0) -- (q0);
		    \draw[transition] (q0) -- node[above] {$\$$} (q1);
		    \draw[transition] (q1) edge[bend left=30] node[above] {$\sigma \in \Sigma$} (q2);
		    \draw[transition] (q2) edge[loop above] node {$\sigma \in \Sigma$} (q2);
		    \draw[transition, Red] (q2) to[bend left=30] node[above] {\textcolor{Black}{$\$$}} node[below] {$ G_\$$} (q1);
		\end{tikzpicture}}
		\caption{Automaton $U$}
    \end{subfigure}%
    ~ 
    \begin{subfigure}[t]{0.49\textwidth}
        \centering
        \resizebox{0.99\textwidth}{!}{
        \begin{tikzpicture}[>=stealth, semithick]
		    \fill[rounded corners, Red, opacity = 0.2] (2,0.2) rectangle (6,2.8);
		    \draw[thick, rounded corners, Red] (2,0.2) rectangle (6,2.8);
		    \node[Red] at (4, 1.5) {\Large $V$};

		    \node[state,Red,fill=white] (p1) at (2, 1.5) {$p_1$};

		    \node[state,Red,fill=white] (p0) at (0, 1.5) {$p_0$};
		    
			\draw[transition,Red] (-0.7,1.5) -- (p0);

		    \draw[transition,Red] (p0) edge[loop above] node[above] {$\sigma \in \Sigma \cup \{\$\}$} (p0);
		    \draw[transition,Red] (p0) -- node[above] {$\$$} (p1);
		\end{tikzpicture}}
		\caption{Cover-VASS $G_\$$}
    \end{subfigure}
    \caption{The \automata{DFA}{Cover-VASS} recognising the delimited Kleene star.}
    \label{fig:kleene-star}
\end{figure*}
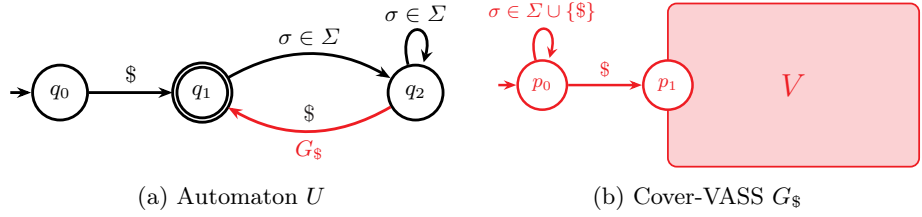

Let $p_1$ be the initial state of $V$. We define $G_\$$ to be $V$ with a new initial state $p_0$ ($p_1$ is no longer initial), plus the following transitions, each having zero effect on the counters:
\vspace{-\topsep}\begin{itemize}
	\item $p_0 \xrightarrow{\sigma} p_0$ for every $\sigma \in \Sigma\cup\set{\$}$, and
	\item $p_0 \xrightarrow{\$} p_1$.
\end{itemize}

We can observe that $G_{\$}$ accepts exactly words of the form $w_1 \$ w_2 \$ \cdots \$ w_k$ such that $k \geq 1$ and $w_k \in L$, that is the last block over $\Sigma$ is in $L$. 
Note that $w_1$ may be empty.

Finally, we can conclude that $A$ accepts $\$$ (by the run $q_0 \$q_1$), and $L^{\$*} = \{\$w_1\$\cdots\$w_k\$ \mid k\ge1 \text{ and for all i}\le k: w_i \in L \}$, as upon reading each $\$$ on a transition $(q_2,\$,q_1)$, the proceeding block must be in $L$, and so every $w_i\in L$.\qed
\end{proof}

\cref{thm:cover-vass-kleene,claim:delkleenestar} together prove~\cref{thm:vassnotcontained}.

\begin{credits}
\section*{\ackname}
We gratefully acknowledge the TCS@Liverpool Open Problems Day for initiating the collaboration. 
This work is supported by the EPSRC project EP/X042596/1 and by the ERC grant INFSYS, agreement no. 950398.
We would like to thank Georg Zetzsche for a discussion about the closure of VASS languages under Kleene star (\cref{subsec:non-deterministic-cover-vass}).
\end{credits}

\bibliographystyle{splncs04}

\appendix

\section{Missing Proofs from Section~\ref{sec:expressivity}}\label{appendix:expressivity}

\subsection{The Expressivity of Regular Guards}
\label{appendix:regularexpressivity}
\thmregexpsize*
Formally, let $A = (\Sigma, U, \Gg)$ be a given \automata{NFA}{NFA} where $U = (\Sigma, Q, \delta, q_0, F)$ is the underlying NFA with transitions $\delta = \set{t_1, \ldots, t_m}$ and $\Gg = \set{L_1, \ldots, L_m}$ are the transition guards.
These transition guards are specified by NFAs $G_1, \ldots, G_m$ (where $\lang{G_i} = L_i$).

First, we shall convert the NFA guards into DFAs.
Suppose that $G_i$ has $n_i$ states.
We shall use the standard powerset construction on the NFA guard $G_i$ to produce a DFA $D_i = (\Sigma, Q_i, \delta_i, q_{i,0}, F_i)$ such that, for every $i \in \set{1, \ldots, m}$, $\lang{D_i} = \lang{G_i}$ and $\abs{Q_i} \leq 2^{n_i}$. 
In other words, the DFA guards are at most exponentially larger than the original NFA guards.
Let $A'$ be the \automata{NFA}{DFA} that is $(\Sigma, U, \Gg)$ where $\Gg$ are the same guards as for $A$, except they are instead specified by the DFAs $D_1, \ldots D_m$.
Since, for every $i \in \set{1, \ldots, m}$, $\lang{D_i} = \lang{G_i}$, it is true that $\lang{A'} = \lang{A}$.

Now, we will simultaneously use the standard powerset construction on the underlying automaton $U$ and take the product of the $m$ DFA guards $D_1, \ldots, D_m$ to convert the \automata{NFA}{DFA} $A'$ into an exponential size DFA $D$.
Each state in $D$ is an $(m+1)$-tuple
\begin{equation*}\label{eq:dfa-states}
	(Q', q_1, \ldots, q_m) \in \powerset{Q} \times Q_1 \times \ldots \times Q_m
\end{equation*}
consisting of the current possible states of the underlying NFA $Q' \sset Q$ and the current states of each guard DFA $q_i \in Q_i$. 

Upon reading $\sigma \in \Sigma$, the DFA $D$ can transition from the state $(P', p_1, \ldots, p_m)$ to the state $(Q', q_1, \ldots, q_m)$ if
\begin{itemize}
	\item for every $i \in \set{1, \ldots, m}$, the $i$-th guard transitions from $p_i$ to $q_i$, i.e.\ $(p_i, \sigma, q_i) \in \delta_i$;
	\item for every $q' \in Q'$, there exists a transition in the underlying NFA $t_j = (p', \sigma, q') \in \delta$ such that $p' \in P'$ and the $j$-th guard was in an accepting state, i.e.\ $p_j \in F_j$; and
	\item $Q'$ is a maximal subset of states, i.e.\ there does not exist a state $q \in Q \setminus Q'$ and a transition $t_j = (p', \sigma, q) \in \delta$ such that $p' \in P'$ and $q_j \in F_j$.
\end{itemize}
If these conditions are satisfied, we add the following transition to the DFA $D$:
\begin{equation*}
	(P', p_1, \ldots, p_m) \xrightarrow{\sigma} (Q', q_1, \ldots, q_m).
\end{equation*}

Let $E \sset Q$ be the set of states in $U$ that can be reached from $q_0$ via \hbox{$\varepsilon$-transitions}.
The initial state of $D$ is therefore $(E, q_{1,0}, \ldots, q_{m,0})$.
The accepting states of $D$ are the states such that the underlying automaton is in an accepting state; $(Q', q_1, \ldots, q_m)$ where $Q' \cap F \neq \emptyset$ (and there are no restrictions on the guards' current states $q_i$).

\cref{thm:exponential_size} follows from Claims~\ref{clm:language-equiv-dfa} and~\ref{clm:exp-size-dfa}.

\begin{restatable}{numberedclaim}{equivdfa}\label{clm:language-equiv-dfa}
		The DFA recognises the same language, $\lang{D} = \lang{A}$.
\end{restatable}

\begin{proof}
	Recall that $\lang{A} = \lang{A'}$, so we shall prove that $\lang{D} = \lang{A'}$.

	First, we will prove that $\lang{D} \sset \lang{A'}$.
	Let $w = \sigma_1 \sigma_2 \cdots \sigma_n \in \lang{D}$ be any word and suppose
	\begin{multline*}
		(E, q_{1,0}, \ldots, q_{m,0}) = (Q^{(0)}, q^{(0)}_1, \ldots, q^{(0)}_m) 
		\xrightarrow{\sigma_1} (Q^{(1)}, q^{(1)}_1, \ldots, q^{(1)}_m) 
		\xrightarrow{\sigma_2} \cdots \\
		\cdots \xrightarrow{\sigma_n} (Q^{(n)}, q^{(n)}_1, \ldots, q^{(n)}_m) 
	\end{multline*}
	is the run that witnesses $w \in \lang{D}$.
	This means that the sequence of configurations 
	\begin{multline*}
		(q_0, q_{1,0}, \ldots, q_{m,0}) = (q^{(0)}, q^{(0)}_1, \ldots, q^{(0)}_m) 
		\xrightarrow{\sigma_1} (q^{(1)}, q^{(1)}_1, \ldots, q^{(1)}_m) 
		\xrightarrow{\sigma_2} \cdots \\
		\cdots \xrightarrow{\sigma_n} (q^{(n)}, q^{(n)}_1, \ldots, q^{(n)}_m) 
	\end{multline*}
	is a run in $A'$ for any choice of $q^{(1)}, \ldots, q^{(n)}$ such that, for every $i \in \set{1, \ldots, n}$, $(q^{(i-1)}, \sigma_{i-1}, q^{(i)}) = t_j \in \delta$, $q^{(i-1)}_j \in F_j$, and $q^{(i)} \in Q^{(i)}$.
	At the final step, since $(Q^{(n)}, q^{(n)}_1, \ldots, q^{(n)}_m)$ is an accepting state of $D$, we know that $Q^{(n)} \cap F \neq \emptyset$, so one can choose $q^{(n)} \in Q^{(n)}$ such that $q^{(n)} \in F$.
	Hence $w \in \lang{A'}$ as well.

	Second, we will prove that $\lang{A'} \sset \lang{D}$.
	Consider any word $w = \sigma_1 \sigma_2 \cdots \sigma_n \in A'$ and suppose 
	\begin{equation*} 
		C_0 
		\xrightarrow{\sigma_1} C_1
		\xrightarrow{\sigma_2} C_2 
		\xrightarrow{\sigma_3} \cdots
		\xrightarrow{\sigma_n} C_n
	\end{equation*}
	is a run in $A'$ where $C_i \in Q \times Q_1 \times \ldots \times Q_m$ is the configuration of $A'$ reached after reading the prefix $\sigma_1 \cdots \sigma_i$; let $C_i = (q^{(i)}, q^{(i)}_1, \ldots, q^{(i)}_m)$.

	We will now find a run in $D$ that witnesses $w \in \lang{D}$.
	To do this, we will inductively identify a state $(Q^{(i+1)}, q^{(i+1)}_1, \ldots, q^{(i+1)}_m)$ of $D$ given the previous state $(Q^{(i)}, q^{(i)}_1, \ldots, q^{(i)}_m)$ of $D$ and the next configuration $C_{i+1}$ of the run in $A'$.
	We will also prove that, for every $i \in \set{0, 1, \ldots, n}$, $q^{(i)} \in Q^{(i)}$.
	
	For the base case, the initial configuration of $D$ is $(E, q^{(0)}_1, \ldots, q^{(0)}_m)$ where $E \sset Q$ is the maximal set of states that can be reached from $q_0$ in $U$ via \hbox{$\varepsilon$-transitions}.
	Observe that $q^{(0)} \in E$, as $q^{(0)} = q_0$ is the starting state of $U$.
	Note that $q^{(0)}_j = q_{j,0}$ are the starting states of the guards. 

	For the inductive step, let $Q^{(i)} \sset Q$ be the maximal set of states that can be reached (in $A'$) after reading the prefix $\sigma_1 \cdots \sigma_i$ and assume that $q^{(i)} \in Q^{(i)}$.
	We shall now define $Q^{(i+1)} \sset Q$ to be the maximal set of states such that, for every state $q' \in Q^{(i+1)}$, there exists a transition $t_j = (p', \sigma_{i+1}, q') \in \delta$ such that $p' \in Q^{(i)}$ and $q^{(i)}_j \in F_j$ (the $j$-th guard was in an accepting state).
	Now, to see why $q^{(i+1)} \in Q^{(i+1)}$, observe that $C_i \xrightarrow{\sigma_{i+1}} C_{i+1}$ tell us that there is a transition $t_j = (q^{(i)}, \sigma_{i+1}, q^{(i+1)}) \in \delta$, where $q^{(i)} \in Q^{(i)}$ and $q^{(i)}_j \in F_j$. 

	Overall, we have identified the run
	\begin{equation*}
		(E, q^{(0)}_1, \ldots, q^{(0)}_m)
		\xrightarrow{\sigma_1} (Q^{(1)}, q^{(1)}_1, \ldots, q^{(1)}_m) 
		\xrightarrow{\sigma_2} \cdots
		\xrightarrow{\sigma_n} (Q^{(n)}, q^{(n)}_1, \ldots, q^{(n)}_m)
	\end{equation*}
	in $D$ such that, for every $i \in \set{0, 1, \ldots, n}$, $q^{(i)} \in Q^{(i)}$.
	Since $w \in \lang{A'}$, we know that $q^{(n)} \in F$, so $Q^{n} \cap F \neq \emptyset$.
	This means that $(Q^{(n)}, q^{(n)}_1, \ldots, q^{(n)}_m)$ is an accepting state of $D$, so $w \in \lang{D}$ as well.
\qed\end{proof}

\begin{numberedclaim}\label{clm:exp-size-dfa}
	The DFA $D$ has at most $2^{\abs{Q} + \sum_{i=1}^m\abs{n_i}}$ states.
\end{numberedclaim}

\begin{proof}
	Let $n$ denote the number of states in~$A$.
	Recall that $n_i$ was the number of states in the original NFA guards and that $\abs{Q_i} \leq 2^{n_i}$.
	Each state in $D$ is an $(m+1)$ tuple in $\powerset{Q} \times Q_1 \times \ldots \times Q_m$. The number of states in $D$ is therefore at most
	\(
		2^{\abs{Q}} \cdot \abs{Q_1} \cdot \ldots \cdot \abs{Q_m} 
		\leq 2^{\abs{Q}} \cdot 2^{n_1} \cdot \ldots \cdot 2^{n_m}
		= 2^{\abs{Q} + \sum_{i=1}^m\abs{n_i}}
	\).
	\qed
\end{proof}

\subsection{Computing Intersections}

\intersectiondfa*
\label{appendix:dfaintersection}
\begin{proof}
	This proof follows a similar construction to the proof of~\cref{pro:intersection-to-NFA}. 
	The transitions from $q_{i-1}$ to $q_i$ are now $\$$-transitions.
	One must also slightly modify the guards.
	Suppose at the initial state, the word $w \in \Sigma^*$ is read and the first guard has been passed.
	Now the automaton has read the word $w\,\$$. 
	This means that the second guard has to be modified so that the guard language is $\set{w\,\$ : w \in \lang{G_1}}$; this can easily be achieved by adding a new final state to $G_1$ that can be reached by reading a $\$$ from any of the original final states (which are no longer final states in the modified automaton).
	This idea needs to be repeated for each guard.
	The $i$-th guard $G_i$ is modified so that the guard language is $\set{w\,\$^i : w \in \lang{G_i}}$; this can be easily achieved by adding $i$ new transitions.
	The size of the \automata{DFA}{X} $A$ is $k+2 + \sum_{i=1}^k \abs{G_i} + i = \Oh(k^2) + \sum_{i=1}^k \abs{G_i}$.

	It is straightforward to see that, for every $w \in \Sigma^*$, $w\,\$^k$ is accepted by $A$ if and only if $w \in \lang{G_i}$ for every $1 \leq i \leq k$.
\qed\end{proof}
 
\section{Missing Proofs from Section~\ref{sec:succinctness}}\label{appendix:succinctness}

\claimexpna*

\begin{proof}
	Let $N$ be an NFA such that $\lang{N} = \set{w \$^k : w \in \bigcap_{i=1}^k L(C_{p_i})}$.
	Let $\ell = p_1 p_2 \cdots p_k$.
	It is true that $\lang{N} = \set{a^{\ell n} \$^k: n \in \NN}$.
	Now, assume for the sake of contradiction that $N$ has less than $\ell$ states.
	Given that $\lang{N}$ contains words with arbitrarily many `$a$'s, there exists a path reading $a^x$, for some $x \geq 0$, from the initial state to a state $q$; there is also a simple cycle at state $q$ reading $a^y$, for some $y \geq 1$; and there is a path from $q$ to a final state reading $a^z\$^k$, for some $z \geq 0$.
	It must be true that $a^x\,a^z\$^k \in \lang{N}$ (take the two paths that end and start at $q$) and $a^x\,a^y\,a^z\$^k \in \lang{N}$ (take the two paths that end and start at $q$ and one iteration of the cycle at $q$).
	This also means that, for some $n \in \NN$, $x+z = \ell n$.
	Since $a^y$ is read by a simple cycle and we have assumed that $N$ has less than $\ell$ states, then it must be true that $1 \leq y < \ell$.
	Thus $\ell n = x+z < x+y+z$ and $x+y+z = \ell n + y < \ell n +\ell = \ell(n+1)$, so $\ell n < x + y + z < \ell(n+1)$.
	We therefore arrive at a contradiction because $a^{x+y+z}\$^k \in \lang{N}$ and $x+y+z$ is not a multiple of $\ell$.

	We conclude that any NFA $N$ such that $\lang{N} = \set{w \$^k : w \in \bigcap_{i=1}^k L(C_{p_i})}$ has at least $p_1 p_2 \cdots p_k > 2^k$ many states, and so its size is at least $2^k$.
\qed\end{proof}

\succinctness*
\begin{proof}
	Let $k$ be an arbitrary natural number and let $p_1, \ldots, p_k$ be the first $k$ primes.
	Consider the \automata{DFA}{DFA} $A_k$ that is obtained by applying~\cref{pro:intersection-to-DFA} to the simple cycle DFAs $C_{p_1}, \ldots, C_{p_k}$.
	This \automata{DFA}{DFA} $A$ recognises the language $\set{w \$^k : w \in \bigcap_{i=1}^k L(C_{p_i})}$ and has size $\Oh(k^2) + \sum_{i=1}^k \abs{C_{p_i}} \leq Ck^2\log(k)$ for some constant $C \geq 1$.
	The last inequality follows from the prime number theorem (see, e.g., \cite{A1998}).

	However, by Claim~\ref{clm:exp-nfa}, any NFA that recognises the language $\set{w \$^k : w \in \bigcap_{i=1}^k L(C_{p_i})}$ has size at least $2^k$.
	Since $\abs{A_k} \leq Ck^2\log(k) \leq Ck^3$ and $\abs{N} \geq 2^k$,
	we deduce that $\frac{\abs{A_k}^{\frac{1}{3}}}{C} \geq k$, so we conclude that $N$ is exponentially larger than $A_k$, $\abs{N} \geq 2^{\frac{\abs{A_k}^{\frac{1}{3}}}{C}}$.
\qed\end{proof}

\nestingnomore*

\begin{proof}
	Proof by induction on the nesting depth $k$.
	The base case, $k = 1$, is exactly~\cref{thm:exponential_size} (see, in particular, Claim~\ref{clm:exp-size-dfa}).

	For the inductive step, let $k > 1$. 
	Assume that for all $A \in \Nn_k$ with $n$ states in total, there exists a DFA $D$ such that $\lang{D} = \lang{A}$ and the number of states in $D$ is at most $2^n$.
	Now, let $A' \in \Nn_{k+1}$.
	Suppose that $A'$ has $n'$ many states in total.
	We would like to prove that there exists a DFA $D'$ such that $\lang{D'} = \lang{A'}$ and the number of states in $D'$ is at most $2^{n'}$.

	Since $A' \in \Nn_{k+1}$, we know that $A'$ consists of an underlying NFA $U$ with guards $G_1, \ldots, G_m \in \Nn_k$.
	Suppose that $U$ has $u$ many states and suppose that, for every $i \in \set{1, \ldots, m}$, $G_i$ has $g_i$ many states in total.
	We know that $n' = u + g_1 + \ldots + g_m$.
	By the inductive assumption, for every $i \in \set{1, \ldots, m}$, we know that there exists a DFA $D_i$ such that $\lang{D_i} = \lang{G_i}$ and the number of states in $D_i$ is at most $2^{g_i}$.
	Let $d_i \leq 2^{g_i}$ be the number of states in $D_i$.
	Now, we have constructed an intermediate \automata{NFA}{DFA}; we will use this to construct a single DFA in the same way as was used in the proof of~\cref{thm:exponential_size} (see the text defining $A$ on Page~\pageref{eq:dfa-states}).
	Accordingly, we obtain a DFA $D'$ with states that are $(m+1)$ tuples $(Q', q_1, \ldots, q_m)$ where $Q'$ is a subset of states of the underlying NFA $U$ and $q_1, \ldots, q_m$ are states of each of the guard DFAs $D_1, \ldots, D_m$, respectively. 
	Suppose the DFA $D'$ has $d'$ many states; then
	\begin{equation*}
		d' \leq 2^u \cdot d_1 \cdot \ldots \cdot d_m \leq 2^u \cdot 2^{g_1} \cdot \ldots \cdot 2^{g_m} = 2^{u + g_1 + \ldots + g_m} = 2^{n'}.
	\end{equation*}
	Lastly, the DFA $D'$ has the same language as the intermediate \automata{NFA}{DFA} (See Claim~\ref{clm:language-equiv-dfa} for details), so $\lang{D'} = \lang{A'}$.
\qed\end{proof}
 
\section{Missing Proofs from Section~\ref{sec:games}}\label{appendix:games}

\gamegraph*
\begin{proof}
	Consider an HCS reachable game specified by an \automata{NFA}{DFA} $A = (\Sigma, U, \Gg)$ where $U = (\Sigma, Q, \delta, q_0, F)$ and a predicate $\owner: Q \to \set{0, 1}$.
	Note that the winning states, for Player~0, are $F \sset Q$.
	Suppose that $\delta = \set{t_1, \ldots, t_m}$ and the guards $\Gg = \set{L_1, \ldots, L_m}$ are specified by DFAs $G_i = (\Sigma, Q_i, \delta_i, q_{0,i}, F_i)$.
	We will construct an exponentially sized graph with an equivalent reachability game.

	The states of the game graph are $V = Q \times Q_1 \times \ldots \times Q_m$.
	The owner of a state is just the owner of the underlying NFA state in the product; let $\owner' : V \to \set{0,1}$ be defined by $\owner'((q, q_1, \ldots, q_m)) = \owner(q)$.
	The winning states for Player~0 in the game graph are those which contain winning states in the underlying NFA for Player~0, $V' = F \times Q_1 \times \ldots \times Q_m$.
	Lastly, the edges of the graph are just the allowed transitions in the product of the underlying NFA and all the DFA guards.
	There is an edge from $(p, p_1, \ldots, p_m)$ to $(q, q_1, \ldots, q_m)$ if there is $\sigma \in \Sigma$ such that $t_j = (p, \sigma, q) \in \delta$, $p_j \in F_j$, and, for every $i \in \set{1, \ldots, m}$, $(p_i, \sigma, q_i) \in \delta_i$.
	The initial state of the reachability game on $(V,E)$ is $(q_0, q_{0,1}, \ldots, q_{0,m})$.

	Clearly, the reachability game on $(V, E)$ faithfully simulates the reachability game on $A$.
	This means that Player~0 wins the reachability game on $(V,E)$ if and only if Player~0 wins the reachability game on $A$.

	It remains to argue that the graph is exponentially sized but only in the number of guards.
	Since $V = Q \times Q_1 \times \ldots \times Q_m$, we know that 
	\begin{align*}
		\abs{V} = \abs{Q} \cdot \abs{Q_1} \cdot \ldots \cdot \abs{Q_m}
		& \leq \abs{Q} \cdot \left( \max_{1 \leq i \leq m} \abs{Q_i} \right)^m \\
		& \leq \abs{Q} \cdot (\abs{Q_1}+\ldots+\abs{Q_m})^m.
	\end{align*}
	Similarly, there is at most one edge in $E$ for each combination of transitions $t \in \delta$, $t_1 \in \delta_1, \ldots,$ and $t_m \in \delta_m$, so
	\begin{align*}
		\abs{E} \leq \abs{\delta} \cdot \abs{\delta_1} \cdot \ldots \cdot \abs{\delta_m}
		& \leq \abs{\delta} \cdot \left( \max_{1 \leq i \leq m} \abs{\delta_i} \right)^m \\
		& \leq \abs{\delta} \cdot (\abs{\delta_1} + \ldots + \abs{\delta_m})^m.
	\end{align*}
\qed
\end{proof}
 
\section{Missing Proofs from Section~\ref{sec:vass}}\label{appendix:vass}

\subsection{Reach-VASS}
\label{appendix:reach-vass}

\reachvassundecidable*

\begin{proof}
	We can assume, without loss of generality, that for reachability in two-counter machines, the initial and target counter values are both $(0,0)$. 
	That is because we can add a series of transitions adding and subtracting whatever non-zero initial and target counter values are desired, respectively.
	So, consider an arbitrary instance of reachability in a two-counter machine with an initial configuration $(p,(0,0))$ and a target configuration $(q, (0,0))$.
	We first make one slight modification, for sake of convenience later in the reduction.
	We add two new states $r$ and $s$ and two transitions $(q, \texttt{zero}_1, r), (r, \texttt{zero}_2, s)$.
	Clearly $(p,(0,0))$ can reach $(q, (0,0))$ if and only if $(p, (0,0))$ can reach $(s, (0,0))$. 
	Fix $M = (Q, T)$ to be the two-counter machine with this modification.
	
	We shall construct an \automata{DFA}{Reach-1-DVASS} $D = (A, U, \Gg)$ which uses the six-letter alphabet that is the set of actions $A$ of the two-counter machine.
	The goal of our construction is such that the words in $\lang{D}$ exactly correspond to the sequence of actions taken by $M$ on a run from $(p, (0,0))$ to $(s, (0,0))$.
	Hence $D$ has a non-empty language if and only if $(p, (0,0))$ can reach $(s, (0,0))$ in $M$.
	For an example of our construction, see~\cref{fig:reach-vass-example} on Page~\pageref{fig:reach-vass-example}. 

	We define the underlying DFA $U = (A, Q, T, p, \set{s})$ to have the same states and transitions as $M$, the initial state is $p$ (the initial state of the reachability instance), and the only accepting state is $s$ (the target set of the reachability instance).
	We know that $U$ is a DFA because of the action determinacy condition: every state $p \in Q$ and for every action $\alpha \in A$, there is at most one transition of the form $(p, \alpha, q)$. 

	Now, we shall define the guards $\Gg$.
	There will only be two distinct non-trivial guard languages in $\Gg$.
	One of the non-trivial guard languages is defined by the 1-VASS $G_1$ and the other is defined by the 1-VASS $G_2$.
	The 1-VASS $G_1$ will consist of one state $g$, that is an accepting state, and will track the value of the first counter.
	There are six transitions in $G_1$: $(g, \texttt{inc}_1, 1, g)$, $(g, \texttt{dec}_1, -1, g)$, and for every $\alpha \in A \setminus \set{\texttt{inc}_1, \texttt{dec}_1}, (s, \alpha, 0, s)$.
	The 1-VASS $G_2$ is the symmetric version of $G_1$, instead it tracks the second counter.
	Lastly, we detail where the non-trivial guards are located.
	Every zero-testing transition $(p, \texttt{zero}_i, q)$ will be guarded by $G_i$.
	This means that, in $U$, to take the transition $(p, \texttt{zero}_i, q)$, the guard $G_i$ must be accepting; this only happens when $G_i$ has counter value $0$ (which is exactly what $\texttt{zero}_i$ intends to test).

	It is straightforward to see that a word $w \in A^*$ that is accepted by $D$ corresponds to a sequence of actions on a run from $(p, (0,0))$ to $(s, (0,0))$ in $M$; we formalise this in Claim~\ref{clm:reach-vass-correctness}.
	The only catch is that in order to reach the only accepting state $s$ of $D$, then one needs to pass the final two zero-testing transitions $(q, \texttt{zero}_i, r)$ and $(r, \texttt{zero}_2, s)$ 
	which (i) forces the two counters to have zero value at the end of the run, and (ii) verifies that neither counter dropped below zero by taking a disallowed $\texttt{dec}_i$ transition (since then the run in $G_1$ or $G_2$ would terminate before the final zero test was taken).

	We can now conclude the proof of~\cref{thm:reach-vass-undecidable} as it follows from Claim~\ref{clm:reach-vass-correctness}: $\lang{D}$ is not empty if and only if there is a run from $(p, (0,0))$ to $(s, (0,0))$ in $M$. 

	\begin{numberedclaim}\label{clm:reach-vass-correctness}
		Let $w \in A^*$ be a sequence of two-counter machine actions, $w \in \lang{D}$ if and only if there is a run in $M$ from $(p, (0,0))$ to $(s, (0,0))$ whose sequence of actions is $w$.
	\end{numberedclaim}	
	\begin{proof}
		Let $w = \alpha_1\alpha_2\cdots\alpha_n \in A^*$ be a sequence of two-counter machine actions such that $w \in \lang{D}$.
		Consider the sequence of transitions $S = ( (q_{i-1}, \alpha_i, q_i) )_{i=1}^n$ that witnesses $w \in \lang{D}$.
		This means that $q_0 = p$ is the starting state and $q_n = s$ is the one accepting state of $D$.
		Let $x_i$ and $y_i$ be the counter values of the 1-VASS guards $G_1$ and $G_2$, respectively, after the first $i$ transitions have been taken.
		We will prove by induction that $(q_i, (x_i, y_i))$ is the configuration of $M$ that is reached after taking the first $i$ transitions in $M$.
		
		Initially, for the base case: $(q_0, (x_0, y_0)) = (p, (0,0))$.
		Now, suppose that after the first $i$ transitions have been taken, $M$ arrives in the configuration $(q_i, (x_i, y_i))$.
		We will show that $(q_i, (x_i, y_i)) \xrightarrow{(q_i, \alpha_{i+1}, q_{i+1})} (q_{i+1}, (x_{i+1}, y_{i+1}))$ in $M$.
		The next transition is $(q_i, \alpha_{i+1}, q_{i+1})$.
		Without loss of generality, suppose that $\alpha_{i+1} \in \set{\texttt{inc}_1, \texttt{dec}_1, \texttt{zero}_1}$.
		
		First, if $\alpha_{i+1} = \texttt{inc}_1$, then the guard $G_1$ increments its counter $x_{i+1} = x_i + 1$ and $y_{i+1} = y_i$ and clearly in $M$:
		\begin{equation*}
			(q_i, (x_i, y_i)) \xrightarrow{(q_i, \texttt{inc}_1, q_{i+1})} (q_{i+1}, (x_i+1, y_i)) = (q_{i+1}, (x_{i+1}, y_{i+1})).
		\end{equation*}
		
		Second, if $\alpha_{i+1} = \texttt{dec}_1$, then the guard $G_1$ decrements its counter if it was at least 1, so $x_i \geq 1$, $x_{i+1} = x_i - 1$, and $y_{i+1} = y_i$ and clearly in $M$: 
		\begin{equation*}
			(q_i, (x_i, y_i)) \xrightarrow{(q_i, \texttt{dec}_1, q_{i+1})} (q_{i+1}, (x_i-1, y_i)) = (q_{i+1}, (x_{i+1}, y_{i+1})).
		\end{equation*}
		
		Third and finally, if $\alpha_{i+1} = \texttt{zero}_1$, then the guard $G_1$ must have been in the accepting configuration $(g_1, 0)$, so $x_i = 0$, so $x_{i+1} = x_i$ and $y_{i+1} = y_i$ and clearly in $M$:
		\begin{equation*}
			(q_i, (x_i, y_i)) \xrightarrow{(q_i, \texttt{zero}_1, q_{i+1})} (q_{i+1}, (x_i, y_i)) = (q_{i+1}, (x_{i+1}, y_{i+1})).
		\end{equation*}

		Now we know that after $n$ transitions, $M$ will arrive in the configuration $(q_n, (x_n, y_n))$.
		Since $q_n = s$ is the accepting configuration, we know that $q_{n-1} = r$ and $q_{n-2} = q$. 
		This is because the only transition that leads to $s$ is $(r, \texttt{zero}_2, s)$ and the only transition that leads to $r$ is $(q, \texttt{zero}_1, r)$.
		This also means that in order to traverse these two transitions the guards $G_1$ and $G_2$ must be in configurations $(g_1, 0)$ and $(g_2, 0)$, respectively. 
		This means that $x_n = 0$ and $y_n = 0$.
		Hence, we conclude that there is a run from $(p, (0,0))$ to $(s, (0,0))$ and its sequence of actions is $w$.

		The reverse direction is very similar. Let $((q_i, (x_i, y_i)))_{i=0}^n$ be a sequence of configurations in $M$ such that $(q_0, (x_0, y_0)) = (p, (0,0))$, $(q_n, (x_n, y_n)) = (s, (0,0))$ and, for every $i \in \set{1, \ldots, n}$, there is the transition $t_i \in T$ such that $(q_{i-1}, (x_{i-1}, y_{i-1})) \xrightarrow{t_i} (q_{i}, (x_{i}, y_{i}))$.  
		Suppose that $t_i = (q_{i-1}, \alpha_i, q_i)$ where $\alpha_i \in A$.
		We would like to prove that $w = \alpha_1\alpha_2\cdots\alpha_n \in \lang{D}$.
		We will argue by induction that after the first $i$ transitions have been taken, the underlying DFA is in state $q_i$, the 1-VASS guard $G_1$ has configuration $(g_1, x_i)$, and the 1-VASS guard $G_2$ has configuration $(g_2, y_i)$.
		
		The base case is trivial: after zero transitions have been taken $M$ is in configuration $(p, (0,0))$, $U$ is at state $p$, $G_1$ has configuration $(g_1, 0)$, and $G_2$ has configuration $(g_2, 0)$.
		Now, assume that after $i$ transitions have been taken, $U$ is at state $q_i$, $G_1$ has configuration $(g_1, x_i)$, and $G_2$ has configuration $(g_2, y_i)$.
		We will show that after the $(i+1)$-st transition is taken in $M$, $(q_i, (x_i, y_i)) \xrightarrow{t_{i+1}} (q_{i+1}, (x_{i+1}, y_{i+1}))$, that $U$ is in state $q_{i+1}$, $G_1$ has configuration $(g_1, x_{i+1})$, and $G_2$ has configuration $(g_2, y_{i+1})$.
		Without loss of generality, assume that $\alpha_{i+1} \in \set{\texttt{inc}_1, \texttt{dec}_1, \texttt{zero}_1}$.
		This means that $y_{i+1} = y_i$ and the configuration of $G_2$ is not changed as the transition taken in $G_2$ is $(g_2, \alpha_{i+1}, 0, g_2)$.
		It only remains to check whether $U$ can take the transition $(q_i, \alpha_{i+1}, q_{i+1})$ and to check what the next configuration of $G_1$ is.

		First, if $\alpha_{i+1} = \texttt{inc}_1$, then $G_1$ transitions, with $(g_1, \texttt{inc}_1, 1, g_1)$, to $(g_1, x_i+1) = (g_1, x_{i+1})$.
		In this case the DFA $U$ freely takes the transition $(q_i, \texttt{inc}_1, q_{i+1})$ to get to $q_{i+1}$.
		Second, if $\alpha_{i+1} = \texttt{dec}_1$, then we know that $x_i \geq 1$ and so $G_1$ transitions, with $(g_1, \texttt{dec}_1, -1, g_1)$, to $(g_1, x_i-1) = (g_1, x_{i+1})$.
		In this case the DFA $U$ also freely takes the transition $(q_i, \texttt{dec}_1, q_{i+1})$ to get to $q_{i+1}$.
		Third and finally, if $\alpha_{i+1} = \texttt{zero}_i$, then the configuration of $G_1$ does not change as the transition $(g_1, \texttt{zero}_1, 0, g_1)$ is taken and we know that $(g_1, x_i) = (g_1, x_{i+1})$.
		We also know that $x_i = 0$; this means that the guard $G_1$ was in the accepting configuration $(g_1, 0)$ which allows $U$ to take the guarded transition $(q_i, \texttt{zero}_1, q_{i+1})$ to $q_{i+1}$.

		Now, we know that after all $n$ transitions have been taken, then $U$ is at state $q_n$, $G_1$ is at configuration $(g_1, x_n)$, and $G_2$ is at configuration $(g_2, y_n)$. 
		Since $q_n = s$ and $s$ is an accepting state, we conclude that $w = \alpha_1\alpha_2\cdots\alpha_n \in \lang{D}$.
	\qed\end{proof}
\end{proof}

\subsection{Cover-DVASS}
\label{appendix:cover_vass}

Let $V = (\Sigma, Q, T, q_0, F)$ be a $d$-DVASS.
In~\cref{subsec:cover-vass}, we stated that $\norm{V}$ is the size of $V$ encoded in unary.
To be explicit, 
\begin{equation*}
	\norm{V} \coloneqq \abs{Q} + \sum_{(p, \sigma, \vec{u}, q) \in T} \max\set{1, \abs{\vec{u}[1]}, \ldots, \abs{\vec{u}[d]}}.
\end{equation*}

Recall, we consider a fixed \automata{NFA}{Cover-DVASS} $A = (\Sigma, U, \Gg)$ where $U = (\Sigma, Q, \delta, q_0, F)$ is the underlying NFA with transitions $\delta = \set{t_1, \ldots, t_m}$ and guards $\Gg = \set{ L_1, \ldots, L_m}$.
For every $i \in \set{1, \ldots, m}$, the guard is specified as a $d_i$-dimensional Cover-DVASS $G_i = (\Sigma, Q_i, T_i, q_{0,i}, F_i)$, so $\lang{G_i} = L_i$.
We use $\norm{V}$ to denote the size of a $d$-VASS encoded in unary.

\existsequivalentvass*
\begin{proof}
	Without loss of generality, we assume that each guard $G_i$ is complete and does not die. 
	If this does not hold, we can just add a non-accepting sink state that is reachable from every state and by reading any symbol.

	From the underlying NFA $U = (\Sigma, Q, \delta, q_0, F)$ and the $d_i$-dimensional Cover-DVASS guards $G_i = (\Sigma, Q_i, T_i, q_{0,i}, F_i)$, we will create a $D$-dimensional (non-deterministic) VASS $V = (\Sigma, Q', T', q_0', F')$.
	Roughly speaking, $V$ is just the product of $U, G_1, \ldots, G_m$.
	Precisely, the states $Q'$ of $V$ are the $(m+1)$-tuples $(q, q_1, \ldots, q_m) \in Q \times Q_1 \times \ldots \times Q_m$ consisting of the current state $q \in Q$ of the underlying NFA and the current states of each guard $q_i \in Q_i$.
	A configuration of $V$ consists of a state $(p, p_1, \ldots, p_m)$ and a $D$-dimensional vector $(\vec{x}_1, \ldots, \vec{x}_m)$.
	The transitions $T'$ of $V$ are defined as follows.
	Fix $\sigma \in \Sigma$. 
	For all transitions $t_j = (p, \sigma, q) \in \delta$ and, for every $i \in \set{1, \ldots, m}$ such that there is a transition $(p_i, \sigma, \vec{u}_i, q_i) \in T_i$, if $p_j \in F_j$ (i.e.\ $G_j$ is an accepting state), then we will add the following transition to $T'$:
	\begin{equation*}
		(p, p_1, \ldots, p_m) \xrightarrow{\sigma, (\vec{u}_1, \ldots, \vec{u}_m)} (q, q_1, \ldots, q_m).
	\end{equation*}

	We shall identify $s = (q_0, q_{0,1}, \ldots, q_{0,m})$ to be the initial state and we shall add one additional final state $t$ to $Q'$.
	We define $F' = \set{t}$, and we will add additional $\varepsilon$-transitions to $T'$ that lead to $t$.
	For all accepting states in the underlying NFA $f \in F$ and, for every combination of states $q_1 \in Q_1, \ldots, q_m \in Q_m$, we will add the transition
	\begin{equation*}
		(f, q_1, \ldots, q_m) \xrightarrow{\varepsilon, \vec{0}} t.
	\end{equation*}

	The $D$-VASS $V$ is designed to simulate the \automata{NFA}{Cover-DVASS} $A$.
	The state of the VASS $(q, q_1, \ldots, q_m)$ tracks the states of the underlying NFA and all the guards.
	The first $d_1$ counters of $V$ track the $d_1$ counters of $G_1$, the next $d_2$ counters of $V$ track the $d_2$ counters of $G_2$ and so on.

	\cref{pro:existsequivalentvass} is broken down into the following two claims (Claims~\ref{clm:cover-vass-correctness} and~\ref{clm:cover-vass-size}).

	\begin{numberedclaim}\label{clm:cover-vass-correctness}
		$\lang{V} = \lang{A}$.
	\end{numberedclaim}
	\begin{proof}
		First, we will argue that $\lang{V} \sset \lang{A}$.
		Let $w = \sigma^{(1)}\sigma^{(2)}\cdots\sigma^{(\ell)} \in (\Sigma \cup \set{\varepsilon})^*$ be a word such that $w \in \lang{V}$ (it will be convenient for us to leave `$\varepsilon$'s in $w$).
		Let $((s^{(k)}, \vec{x}^{(k)}))_{k=0}^\ell$ be a run in $V$ such that $(s^{(0)}, \vec{x}^{(0)}) = (s, \vec{0})$, $(s^{(\ell)},  \vec{x}^{(\ell)}) = (t, \vec{x})$ for some $\vec{x} \geq \vec{0}$, and, for every $k \in \set{1, \ldots, \ell}$, 
		\begin{equation*}
			(s^{(k-1)},\sigma^{(k)}, \vec{x}^{(k)} - \vec{x}^{(k-1)}, s^{(k)}) \in T'.
		\end{equation*}
		For every $k \in \set{0, 1, \ldots, \ell}$, let $s^{(k)} = (q^{(k)}, q^{(k)}_1, \ldots, q^{(k)}_m)$ and $\vec{x}^{(k)} = (\vec{x}^{(k)}_1, \ldots, \vec{x}^{(k)}_m)$ where, for every $i \in \set{1, \ldots, m}, \vec{x}^{(k)}_i \in \NN^{d_i}$.

		We shall prove, by induction, that there is an equivalent run in $A$.
		Specifically, we will prove that after the first $k$ transitions, reading $\sigma^{(1)}\sigma^{(2)} \cdots \sigma^{(k)}$, have been taken in $V$, there is a run in $A$, reading $\sigma^{(1)}\sigma^{(2)} \cdots \sigma^{(k)}$, to a configuration in which the underlying NFA $U$ is in state $q^{(k)}$, and for every $i \in \set{1, \ldots, m}$, the configuration of $G_i$ is $(q^{(k)}_i, \vec{x}^{(k)}_i)$.

		The base case is immediate: after zero transitions have been taken $V$ is in the starting configuration $(s, \vec{0}) = ((q_0, q_{0,1}, \ldots, q_{0,m}), (\vec{0}, \ldots, \vec{0}))$.
		The underlying NFA $U$ is at starting state $q_0$ and, for every $i \in \set{1, \ldots, m}$, the starting configuration of $G_i$ is $(q_{0,i}, \vec{0})$.
		In both automata, the empty word $\varepsilon$ has been read.

		For the inductive step, assume that after $k$ transition have been taken, there is a run reading $\sigma^{(1)}\sigma^{(2)} \cdots \sigma^{(k)}$ in $A$ so that the underlying NFA is at state $q^{(k)}$ and, for every $i \in \set{1, \ldots, m}$, the configuration of $G_i$ is $(q^{(k)}_i, \vec{x}^{(k)}_i)$.
		Now, given that $(s^{(k)}, \sigma^{(k+1)},\vec{x}^{(k+1)} - \vec{x}^{(k)}, s^{(k+1)}) \in T'$, we know that there exists a transition $t_j = (q^{(k)}, \sigma^{(k+1)}, q^{(k+1)}) \in \delta$ and we know that $q^{(k)}_j \in F_j$ ($G_j$ is in an accepting state which allows $t_j$ to be taken in $U$).
		It is also true that, for every $i \in \set{1, \ldots, m}$, there exists a transition $(q^{(k)}_i, \sigma^{(k+1)}, \vec{u}_i, q^{(k+1)}_i) \in T_i$ such that $\vec{x}^{(k+1)} - \vec{x}^{(k)} = (\vec{u}_1, \ldots, \vec{u}_m)$.
		Now, after $(q^{(k)}, \sigma^{(k+1)}, q^{(k+1)})$ is taken in $U$, the word $\sigma^{(1)}\sigma^{(2)}\cdots\sigma^{(k)}\sigma^{(k+1)}$ has been read, the state of $U$ will be $q^{(k+1)}$, and for every $i \in \set{1, \ldots, m}$, the configuration of $G_i$ will be $(q^{(k+1)}, \vec{x}^{(k)}_i + \vec{u}^{(k)}_i) = (q^{(k+1)}, \vec{x}^{(k+1)}_i)$.
		This concludes this proof by induction.

		Now, by construction of $V$, we know that after reading $\sigma^{(1)}\sigma^{(2)}\cdots\sigma^{(\ell-1)}$, the last transition in the run must be $(q^{(\ell-1)}, q^{(\ell-1)}_1, \ldots, q^{(\ell-1)}_m) \xrightarrow{\varepsilon, \vec{0}} t$.
		We therefore deduce that $q^{(\ell-1)}$ is an accepting state of $U$ and we have therefore found that $\sigma^{(1)}\sigma^{(2)}\cdots\sigma^{(\ell-1)}\varepsilon = \sigma^{(1)}\sigma^{(2)}\cdots\sigma^{(\ell-1)}\sigma^{(\ell)} = w \in \lang{A}$.

		Second, we will prove that $\lang{A} \sset \lang{V}$.
		Let $w = \sigma^{(1)} \sigma^{(2)} \cdots \sigma^{(\ell)} \in (\Sigma \cup \set{\varepsilon})^*$ be a word such that $w \in \lang{A}$.
		Suppose there is a run witnessing $w \in \lang{A}$ which consists of $\ell+1$ many configurations.
		After $k \leq \ell$ many transitions, suppose that the word $\sigma^{(1)}\sigma^{(2)}\cdots\sigma^{(k)}$ has been read and the \automata{NFA}{Cover-VASS} $A$ arrives at a configuration in which $U$ is at state $q^{(k)}$, and for every $i \in \set{1, \ldots, m}$, the configuration of $G_i$ is $(q^{(k)}_i, \vec{x}^{(k)}_i)$.
		We will prove, by induction, that $( ((q^{(k)}, q^{(k)}_1, \ldots, q^{(k)}_m), (\vec{x}^{(k)}_1, \ldots, \vec{x}^{(k)}_m)) )_{k=0}^{\ell}$ is a run in $V$ starting from $(s, \vec{0})$ that reads $\sigma^{(1)}\sigma^{(2)}\cdots\sigma^{(\ell)}$.

		The base case is straightforward.
		After zero transitions have been taken $U$ is at state $q_0$, and for every $i \in \set{1, \ldots, m}$, the configuration of $G_i$ is $(q_{0,i}, \vec{0})$ so the initial configuration of $G_i$ is $((q_0, q_{0,1}, \ldots, q_{0,m}), (\vec{0}, \ldots, \vec{0})) = (s, \vec{0})$, as required.
		In both automata, the empty word $\varepsilon$ has been read.

		Now, for the inductive step, assume that after $k > 0$ transitions have been taken and $\sigma^{(1)}\sigma^{(2)}\cdots\sigma^{(k)}$ has been read, $V$ arrives at the configuration 
		\begin{equation*}
			\left((q^{(k)}, q^{(k)}_1, \ldots, q^{(k)}_m), (\vec{x}^{(k)}_1, \ldots, \vec{x}^{(k)}_m)\right).
		\end{equation*}
		We know that, in $A$, upon reading the next letter $\sigma^{(k+1)}$, there is $\sigma^{(k+1)}$-transition that takes $U$ from $q^{(k)}$ to $q^{(k+1)}$, and for every $i \in \set{1, \ldots, m}$, there is a $\sigma^{(k+1)}$-transition that takes $G_i$ from $(q^{(k)}_i, \vec{x}^{(k)}_i)$ to $(q^{(k+1)}_i, \vec{x}^{(k+1)}_i)$ as we assumed that $G_i$ does not die.
		Precisely, let $t_j = (q^{(k)}, \sigma^{(k+1)}, q^{(k+1)}) \in \delta$ be the transition taken by $A$, and for every $i \in \set{1, \ldots, m}$, let transition $(q^{(k)}_i, \sigma^{(k+1)}, \vec{x}^{(k+1)}_i - \vec{x}^{(k)}_i, q^{(k+1)}_i) \in T_i$ be the transition taken by $G_i$.
		It is also true that $q^{(k)}_j \in F_j$ (to enable $t_j$ to be taken in $U$).
		Altogether, this means that the following transition is in $T'$:
		\begin{equation*}
			(q^{(k)}, q^{(k)}_1, \ldots, q^{(k)}_m)
			\xrightarrow{\sigma^{(k+1)}, \left(\vec{x}^{(k+1)}_1 - \vec{x}^{(k)}_1, \ldots, \vec{x}^{(k+1)}_m - \vec{x}^{(k)}_m\right)}
			(q^{(k+1)}, q^{(k+1)}_1, \ldots, q^{(k+1)}_m).
		\end{equation*}
		Thus, by taking this transition from $((q^{(k)}, q^{(k)}_1, \ldots, q^{(k)}_m), (\vec{x}^{(k)}_1, \ldots, \vec{x}^{(k)}_m))$, $V$ arrives in the configuration 
		\begin{equation*}
			\left((q^{(k+1)}, q^{(k+1)}_1, \ldots, q^{(k+1)}_m), (\vec{x}^{(k+1)}_1, \ldots, \vec{x}^{(k+1)}_m)\right).
		\end{equation*}

		Now, after all $\ell$ transitions have been taken and $\sigma^{(1)}\sigma^{(2)}\cdots\sigma^{(\ell)}$ has been read, the $D$-VASS $V$ arrives at the configuration $((q^{(\ell)}, q^{(\ell)}_1, \ldots, q^{(\ell)}_m), (\vec{x}^{(\ell)}_1, \ldots, \vec{x}^{(\ell)}_m))$ where $q^{(\ell)} \in F$ is an accepting state of $U$ and, for every $i \in \set{1, \ldots, m}$,\linebreak 
		$\vec{x}^{(\ell)}_i \geq \vec{0}$.
		This means that $V$ can then take one last transition:
		\begin{equation*}
			(q^{(\ell)}, q^{(\ell)}_1, \ldots, q^{(\ell)}_m)
		\xrightarrow{\varepsilon, \vec{0}} t, 
		\end{equation*}
		and arrive at the configuration $(t, (\vec{x}^{(\ell)}_1, \ldots, \vec{x}^{(\ell)}_m))$.
		Since $(\vec{x}^{(\ell)}_1, \ldots, \vec{x}^{(\ell)}_m) \geq \vec{0}$, we have found a run that reads $\sigma^{(1)}\sigma^{(2)}\cdots\sigma^{(\ell)} \varepsilon = \sigma^{(1)}\sigma^{(2)}\cdots\sigma^{(\ell)} = w$ from $(s, \vec{0})$ to a configuration $(t, \vec{x})$, in $V$, for some $\vec{x} \geq \vec{0}$; hence $w \in \lang{V}$.

		Since $\lang{V} \sset \lang{A}$ and $\lang{A} \sset \lang{V}$, we conclude that $\lang{V} = \lang{A}$.
	\qed\end{proof}

	\begin{numberedclaim}\label{clm:cover-vass-size}
		$\norm{V} \leq 2(\norm{U} + \norm{G_1} + \ldots + \norm{G_m})^{m+2}$.
	\end{numberedclaim}
	\begin{proof}
		Recall that, for the $D$-VASS $V = (\Sigma, Q', T', q_0, \set{t})$, its size is
		\begin{equation*}
			\norm{V} = \abs{Q'} + \sum_{(p, \sigma, \vec{u}, q) \in T'} \max\set{1, \abs{\vec{u}[1]}, \ldots, \abs{\vec{u}[D]}}.
		\end{equation*}
		For convenience, we define $\norm{T'} \coloneqq \max\set{1, \abs{\vec{u}[1]}, \ldots, \abs{\vec{u}[D]} : (p, \sigma, \vec{u}, q) \in T'}$.
		This means that 
		\begin{equation}\label{equ:lazy}
			\norm{V} \leq \abs{Q'} + \abs{T'}\cdot\norm{T'}.
		\end{equation}

		For the states, we know that $Q' = Q \times Q_1 \times \ldots \times Q_m \cup \set{t}$, so 
		\begin{align}\label{equ:q'}
			\abs{Q'} 
			= \abs{Q} \cdot \abs{Q_1} \cdot \ldots \cdot \abs{Q_m} + 1 
			& \leq (\abs{Q} + \abs{Q_1} + \ldots + \abs{Q_m})^{m+1} \notag\\
			& \leq (\norm{U} + \norm{G_1} + \ldots + \norm{G_m})^{m+1}.
		\end{align}

		For the transitions, recall that there are two kinds of transitions in~$T'$.
		Transitions of the first kind are the transitions that simulate transitions in $A$ and the second kind of transitions are the transitions that lead to the target state $t$ (these transitions have the update vector $\vec{0}$).
		There is at most one transition of the first kind for every combination of transitions $t \in \delta$, $t_1 \in T_1$, \ldots, $t_m \in T_m$.
		There is at most one transition of the second kind for every state in $Q'$ that is not $t$.
		Thus,
		\begin{align}\label{equ:t'}
			\abs{T'}\cdot\norm{T'} 
			& \leq (\abs{\delta} \cdot \abs{T_1} \cdot \ldots \cdot \abs{T_m} + \abs{Q} \cdot \abs{Q_1} \cdot \ldots \cdot \abs{Q_m}) \cdot \norm{T'} \notag\\
			& \leq \left( (\abs{\delta} + \abs{T_1} + \ldots + \abs{T_m})^{m+1} + (\abs{Q} + \abs{Q_1} + \ldots + \abs{Q_m})^{m+1} \right) \cdot \norm{T'} \notag\\
			& \leq (\abs{Q} + \abs{\delta} + \abs{Q_1} + \abs{T_1} + \ldots + \abs{Q_m} + \abs{T_m})^{m+1} \cdot \norm{T'} \notag\\
			& \leq (\norm{U} + \norm{G_1} + \ldots + \norm{G_m})^{m+1} \cdot \norm{T'} \notag\\
			& \leq (\norm{U} + \norm{G_1} + \ldots + \norm{G_m})^{m+2}.
		\end{align}
		Where the last inequality follows from the fact that $\norm{T'} = \max\set{\norm{T_1}, \ldots, \norm{T_m}}$, so $\norm{T'} \leq \norm{T_1} + \ldots + \norm{T_m} \leq \norm{G_1} + \ldots + \norm{G_m}$.

		By combining Equations~\ref{equ:lazy}, ~\ref{equ:q'} and~\ref{equ:t'} we deduce that
		\begin{align*}
			\norm{V} 
			& \leq (\norm{U} + \norm{G_1} + \ldots + \norm{G_m})^{m+1} + (\norm{U} + \norm{G_1} + \ldots + \norm{G_m})^{m+2} \\
			& \leq 2(\norm{U} + \norm{G_1} + \ldots + \norm{G_m})^{m+2}.
		\end{align*}\qed\end{proof}
	Claims~\ref{clm:cover-vass-correctness} and~\ref{clm:cover-vass-size} complete the proof of \cref{pro:existsequivalentvass}.\qed\end{proof}

\subsection{Cover-VASS}
\label{appendix:non-det-cover-vass}

\covervasskleene*
\begin{proof}
	We prove this by exhibiting a language $L$ that is accepted by a Cover-VASS, but whose Kleene closure $L^*$ is not.

	\begin{numberedclaim}\label{claim:L-is-cover-vass}
	The language $L = \{a^n b^m : 0 \leq m \leq n\}$ is accepted by a Cover-VASS.
	\end{numberedclaim}

	\begin{proof}%
	Consider a VASS with one counter and two states $q_0, q_1$. 
	The initial state is $q_0$ and $q_1$ is the only final state.
	The VASS has three transitions:
	\begin{itemize}
	    \item $q_0 \xrightarrow{a, +1} q_0$ (increment counter on reading $a$),
	    \item $q_0 \xrightarrow{\varepsilon, 0} q_1$ (non-deterministically switch to state $q_1$), and
	    \item $q_1 \xrightarrow{b, -1} q_1$ (decrement counter on reading $b$).
	\end{itemize}
	Starting from $q_0$ with counter value 0, this VASS accepts exactly those words $a^n b^m$ where the final counter value is non-negative, i.e., $n - m \geq 0$.
	\qed\end{proof}

	\claimLstarnotcovervass*

	We are not aware of a direct proof in the literature for Cover-VASS. 
	We prove it using a proof strategy that is similar in nature to the proof of~{\cite[Lemma 3]{BosePT23}}.

	\begin{proof}[Proof of Claim~\ref{claim:Lstar-not-cover-vass}]
	Suppose, for sake of contradiction, that $L^*$ is accepted by some $d$-dimensional Cover-VASS $V$ with state set $Q$, and where all transitions increment or decrement counters by at most $M \in \mathbb{N}$.

	\paragraph{Step 1: Existence of Non-negative Cycles.}

	Consider words of the form:
	\[w_k = (a^{n_1} b^{n_1})(a^{n_2} b^{n_2})(a^{n_3} b^{n_3}),\dots,(a^{n_k} b^{n_k}) \quad \ldots\]
	where $n_i$ is a rapidly growing sequence (to be specified). Since $w_k$ is in $L^{*}$, there is an accepting run on each $w_k$, and thus accepted by $V$. 

	After processing each prefix, we reach some configuration. By Dickson's Lemma, $\mathbb{N}^d$ under component-wise ordering is a well-quasi-order, so there cannot be an infinite antichain of configurations that are all reachable from the initial configuration.

	In particular, consider a word of the form $(a^{n_1}b^{n_1}) \cdots (a^{n_j}b^{n_j}) a^{n_{j+1}}$ where $n_{j+1}$ is chosen sufficiently large. Since the state space is finite and $n_{j+1}$ can be arbitrarily large, during the processing of $a^{n_{j+1}}$, there must exist positions $x < y \leq n_{j+1}$ and a state $q$ such that:
	\begin{itemize}
	    \item after reading $(a^{n_1}b^{n_1}) \cdots (a^{n_j}b^{n_j}) a^{x}$, we reach configuration $\mathbf{x}$ at state $q$, 
	    \item after reading $(a^{n_1}b^{n_1}) \cdots (a^{n_j}b^{n_j}) a^{y}$, we reach configuration $\mathbf{y}$ at state $q$, and
	    \item $\mathbf{y} \geq \mathbf{x}$ (component-wise).
	\end{itemize}

	This is guaranteed by Dickson's Lemma: as we process the $a$'s, if we never revisit the same state with a configuration that is component-wise greater than or equal to a previous configuration at that state, we would generate an infinite antichain in $\mathbb{N}^d$ (this creates a contradiction). 
	The cycle from $(q, \mathbf{x})$ to $(q, \mathbf{y})$ reads only `$a$'s (specifically, $a^{\ell_{j+1}}$) and has a non-negative effect vector $\mathbf{y} - \mathbf{x} \geq \mathbf{0}$.

	\paragraph{Step 2: Characterizing Non-negative Cycles.}

	Henceforth, we work with a fixed choice of block sizes $n_1, n_2, \ldots, n_K$ where $K = |Q| \cdot 2^d + 1$, and each $n_i$ is chosen large enough to guarantee the existence of non-negative cycles as argued above.

	Consider an accepting run on $(a^{n_1}b^{n_1})(a^{n_2}b^{n_2}) \cdots (a^{n_K}b^{n_K})$. For each block $i \in \{1, \ldots, K\}$, within the $a^{n_i}$ portion, there exists at least one non-negative cycle. Specifically, there exist:
	\begin{itemize}
	    \item a state $q_i$,
	    \item configurations $\mathbf{x}_i$ and $\mathbf{y}_i$ both at state $q_i$ within the $a^{n_i}$ portion,
	    \item a path from $(q_i, \mathbf{x}_i)$ to $(q_i, \mathbf{y}_i)$ reading only `$a$'s, and 
	    \item $\mathbf{y}_i \geq \mathbf{x}_i$ (component-wise).
	\end{itemize}

	For each such cycle, we define
	$$S_i \coloneqq \mathrm{support}(\mathbf{y}_i - \mathbf{x}_i) = \{j \in \{1, \ldots, d\} : (\mathbf{y}_i)_j > (\mathbf{x}_i)_j\}.$$
	That is, $S_i$ is the set of counters that are strictly increased by the cycle. (All other counters have zero effect.)

	\paragraph{Step 3: Pigeonhole Argument.}

	The pair $(q_i, S_i)$ takes values in $Q \times 2^{\{1,\ldots,d\}}$, which has cardinality $|Q| \cdot 2^d$. Since we have $K = |Q| \cdot 2^d + 1$ blocks, by the pigeonhole principle, there exist distinct blocks $j < k$ such that:
	$$(q_j, S_j) = (q_k, S_k).$$
	Let $q = q_j = q_k$ and $S = S_j = S_k$; we have:
	\begin{itemize}
	    \item a cycle from $(q, \mathbf{x}_j)$ to $(q, \mathbf{y}_j)$ in block $j$ with $\mathrm{support}(\mathbf{y}_j - \mathbf{x}_j) = S$, and
	    \item a cycle from $(q, \mathbf{x}_k)$ to $(q, \mathbf{y}_k)$ in block $k$ with $\mathrm{support}(\mathbf{y}_k - \mathbf{x}_k) = S$.
	\end{itemize}

	\paragraph{Step 4: Cycle Rearrangement.}

	Since both cycles have the same support $S$, it is true that
	\begin{itemize}
	    \item for all $i \in S$, $(\mathbf{y}_j - \mathbf{x}_j)_i > 0$ and $(\mathbf{y}_k - \mathbf{x}_k)_i > 0$; and
	    \item for all $i \notin S$, $(\mathbf{y}_j - \mathbf{x}_j)_i = 0$ and $(\mathbf{y}_k - \mathbf{x}_k)_i = 0$.
	\end{itemize}

	Choose $m \in \mathbb{N}$ such that $m \cdot (\mathbf{y}_j - \mathbf{x}_j) \geq (\mathbf{y}_k - \mathbf{x}_k)$ component-wise. 
	Such an $m$ exists because, for every $i \in S$, set $m_i = \lceil (\mathbf{y}_k - \mathbf{x}_k)_i / (\mathbf{y}_j - \mathbf{x}_j)_i \rceil$, and take $m = \max_{i \in S} m_i$. For $i \notin S$, both sides are zero, so the inequality holds automatically.

	Now consider the following modification of the accepting run.
	\begin{description}
	    \item[Pump] the cycle in block $j$. 
	    Repeat the cycle from $(q, \mathbf{x}_j)$ to $(q, \mathbf{y}_j)$ a total of $m+1$ times instead of once. This increases the counter values by an \emph{additional} $m \cdot (\mathbf{y}_j - \mathbf{x}_j)$.
	    \item[Depump] the cycle in block $k$.
	    Remove one iteration of the cycle from $(q, \mathbf{x}_k)$ to $(q, \mathbf{y}_k)$. This decreases the counter values by $(\mathbf{y}_k - \mathbf{x}_k)$.
	\end{description}

	The net effect on the counter values is $m \cdot (\mathbf{y}_j - \mathbf{x}_j) - (\mathbf{y}_k - \mathbf{x}_k) \geq \mathbf{0}$.

	Therefore, the modified run remains non-negative throughout (since we only added non-negative vectors to the configuration, and only earlier in the overall run). Hence, any suffix in the original word can still be applied and we still reach a non-negative accepting state.

	\paragraph{Step 5: Deriving the Contradiction.}

	Let $\ell_j$ be the number of `$a$'s read by one iteration of the cycle in block $j$ and $\ell_k > 0$ be the number of `$a$'s read by one iteration of the cycle in block $k$. The modified run accepts a word of the following form.
	$$\cdots (a^{n_j + m\ell_j} b^{n_j}) \cdots (a^{n_k - \ell_k} b^{n_k}) \cdots$$
	This word contains the block $a^{n_k - \ell} b^{n_k}$ where the number of `$b$'s exceeds the number of `$a$'s. 
	Therefore, this word is not in $L^*$. 
	However, by Step 4 we have an accepting run.
	This contradicts the assumption that $V$ accepts exactly $L^*$.
	\end{proof}

	We have shown that $L = \{a^n b^m : 0 \leq m \leq n\}$ is a Cover-VASS language, but $L^*$ is not. Therefore, the class of Cover-VASS languages is not closed under Kleene star.

	\paragraph{Step 6: Not closed under delimited Kleene star.}
	It is clear, adding a fresh `$\$$' symbol between each $a^nb^n$ block does not affect the argument. More generally, if $L^{\$*}=\$  \cup \$ L\$ \cup \$L\$L\$ \cup \cdots$ were recognised by a Cover-VASS $V$, then by changing $\$$ symbols to $\varepsilon$-transitions, the Kleene star $L^*$ would be recognised, contradicting the previous argument. Hence, the class of Cover-VASS is not closed under delimited Kleene star.\qed
\end{proof}

\end{document}